\documentclass[12pt]{article}

\usepackage{amsfonts}
\usepackage{amsmath}
\usepackage{amssymb}
\usepackage{amsthm}
\usepackage{mathrsfs}

\usepackage{comment}
\usepackage{graphicx}

\newcommand{\paren}[1]{\left(#1\right)}

\newcommand{\powerseries}[2]{#1\!\left<\!\left<#2\right>\!\right>}

\newcommand{\PSA}{\powerseries{S}{A}}
\newcommand{\PSAfinite}{\powerseries{S}{A^*}}
\newcommand{\PSAfinitenoempty}{\powerseries{S}{A^*/\{\emptystring\}}}
\newcommand{\PSAinfinite}{\powerseries{S^\nats}{A^\omega}}
\newcommand{\PSAbiinfinite}{\powerseries{S^{\ints\times\nats}}{A^\zeta}}

\newcommand{\PSAdiverging}{\powerseries{S^\nats}{A^\omega}}

\newcommand{\nats}{\mathbb{N}}
\newcommand{\ints}{\mathbb{Z}} 
\newcommand{\cmps}{\mathbb{C}}
\newcommand{\reals}{\mathbb{R}}
\newcommand{\bools}{\mathbb{B}}

\newcommand{\Rat}{\mathfrak{Rat}}
\newcommand{\Rec}{\mathfrak{Rec}}
\newcommand{\Aut}{\mathfrak{Aut}}

\newcommand{\ratnoempty}{\Rat^*_{/\epsilon}(S,A)}

\newcommand{\aut}[1]{\mathcal{#1}}
\newcommand{\ofaut}[2]{#1^{\aut{#2}}}
\newcommand{\ofautwithsup}[3]{#1^{\aut{#2},#3}}
\newcommand{\defsameofaut}[3]{\ofaut{#1}{#2}:=\ofaut{#1}{#3}}
\newcommand{\auttuple}[1]{(\ofaut{Q}{#1},\ofaut{I}{#1},\ofaut{F}{#1},\ofaut{M}{#1})}
\newcommand{\Mofautwithsup}[2]{\ofautwithsup{M}{#1}{#2}}

\newcommand{\otherwise}{\text{otherwise}}

\newcommand{\emptystring}{\epsilon}

\newcommand{\bicoefs}{S^{\ints\times\nats}}

\newcommand{\allprefixes}[2]{\rho^\omega(#1,#2)}
\newcommand{\allbiprefixes}[2]{\rho^\zeta(#1,#2)}

\newcommand{\slice}[3]{#1_{[#2:#3]}}
\newcommand{\charat}[2]{#1_{[#2]}}

\newcommand{\norm}[1]{|#1|}

\newcommand{\behav}[1]{\left\|#1\right\|}

\newtheorem{theorem}{Theorem}
\newtheorem{proposition}{Proposition}
\newtheorem{lemma}{Lemma}
\newtheorem{corollary}{Corollary}

\usepackage{setspace}
\usepackage[affil-it]{authblk}

\begin{document}

\author{Gregory Crosswhite}
\title{Embracing divergence: a formalism for when your semiring is simply not complete, with applications in quantum simulation}
\affil{Department of Physics \\ University of Queensland \\ Brisbane, Australia}
\date{July 30, 2012}

\maketitle

\begin{abstract}
There is a fundamental difficulty in generalizing weighted automata to the case of infinite words: in general the infinite sum-of-products from which the weight of a given word is derived will diverge.  Many solutions to this problem have been proposed, including restricting the type of weights used (see Refs. \cite{Esik2005I}, \cite{Esik2005II}, and \cite{Esik2009}) and employing a different valuation function that forces convergence (see Refs. \cite{Chatterjee2010}, \cite{Droste2006}, \cite{Droste2007a}, and \cite{Arfi2009}).  In this paper we describe an alternative approach that, rather than seeking to avoid the inevitable divergences, instead \emph{embraces} them as a source of useful information.  Specifically, rather than taking coefficients from an arbitrary semiring $S$ we instead take them from $S^\mathbb{N}$.  Doing this is useful because gives us information about \emph{how} the weight of an infinite word does or does not diverge, and if it does diverge what form the divergence takes --- e.g., polynomial, exponential, etc.  This approach has proved to be incredibly useful in the field of \emph{quantum simulation} (see Refs. \cite{Orus2008}, \cite{McCulloch2008} and \cite{Crosswhite2008}) because when studying infinite systems, information about how quantities of interest (such as energy or magnetization) diverge is exactly what we want.

In this paper we introduce a new kind of automaton which we call a \emph{diverging automaton} that maps infinite words to sequences of weights from a semiring and which employs a B\"uchi-like acceptance condition.  We then develop a theory for \emph{diverging power series} and prove a Kleene Theorem connecting \emph{rational} diverging power series to diverging automata.  Afterward we repeat this process by introducing \emph{bidiverging automata} which map \emph{biinfinite} words to elements in $S^{\ints\times\nats}$, developing a theory for \emph{bidiverging power series}, and proving another Kleene Theorem.  We conclude by describing how bidiverging automata are applied to simulate biinfinite quantum systems.
\end{abstract}

\tableofcontents

\section{Introduction}

There is a fundamental difficulty in generalizing weighted automata to the case of infinite words: In general the infinite sum-of-products from which the weight of a given word is derived will diverge.  Of course, one solution is to restrict oneself to the class of \emph{complete} semirings for which an infinite sum-of-products is always guaranteed to converge, thereby excluding the possibility of divergences altogether (see Refs. \cite{Esik2005I}, \cite{Esik2005II}, and \cite{Esik2009} for a description of the resulting theory developed in a very general setting);  this is a perfectly satisfactory approach in settings where one's semirings meet the necessary requirements, but it nonetheless places very restrictive conditions on the semirings one can use.  Another solution is to modify the definition of the weight of a word by allowing it to be an arbitrary function of the infinite sequences of weights along the paths rather than a straightforward sum of the products, and then to pick this function so that it is well-defined for all automata and words (for examples, see Refs. \cite{Chatterjee2010} and \cite{Droste2007a}).  A specific example of the latter approach is the use of a deflation parameter that causes the weights of transitions to exponentially decrease in magnitude over the course of a run through the automata, thus \emph{forcing} divergent sums to converge (see Refs. \cite{Droste2006}, \cite{Chatterjee2010}, and \cite{Droste2007b} for development of the theory and \cite{Arfi2009} for an application of it to game theory).  (Note that the use of a deflation parameter can be viewed as a special case of the first approach; see \S 4 of Ref. \cite{Esik2005II} for the details.)

In this paper we shall describe an alternative approach to handling the infinite sum-of-products resulting from the marriage of weights and infinite words: rather than seeking to avoid the inevitable divergences, we \emph{embrace} them as a source of useful information.  Specifically, we start with an arbitrary semiring $S$, but rather than assigning weights from $S$ to the infinite words, we assign weights from $S^\mathbb{N}$;  doing this is useful because each sequence of values from $S$ gives us information about exactly how the weight of an infinite word does or does not diverge, and if it does diverge what form the divergence takes --- e.g., polynomial, exponential, etc.

This approach has proved to be incredibly useful in the field of \emph{quantum simulation} (see Refs. \cite{Orus2008}, \cite{McCulloch2008} and \cite{Crosswhite2008}, replacing the term ``infinite matrix product state'' with ``bidiverging automaton'' and reducing the level of mathematical formalism expected by the reader), by which is meant the use of \emph{classical} computers to simulate \emph{quantum} systems (as opposed to the hypothetical use of quantum computers for simulation).  The reason for this is as follows:  It is often very useful to study the theoretical properties of infinitely large physical systems because this gives us valuable information about the bulk properties of the system --- that is, the properties of the system when effects due to the presence of the physical boundaries of the system are negligible.  In the infinite setting, physical quantities such as the total energy or the total magnetization of the system that increase as the number of particles increase will all be infinitely large, so merely learning that these quantities are divergent is not useful;  instead, what we really want to know is what the functions of the quantities are with respect to the size of the system.  It is for this reason that working with weights that are \emph{sequences} of values from a semiring supplies exactly what is needed to extract the desired information from a simulation.

The remainder of this paper shall be structured as follows.  In \S\ref{sec:diverging-automata} we shall formally introduce what we shall call a \emph{diverging automaton}, which is a kind of automaton that maps infinite words to infinite sequences of weights from a semiring.  In \S\ref{sec:diverging-power-series} we shall define a kind of power series over infinite words that we shall call a \emph{diverging power series}, as well as rational operators on this series, and then we shall prove a Kleene Theorem connecting rational diverging power series and diverging automata.  In \S\ref{sec:bidiverging-automata} we shall likewise introduce a kind of automaton that maps biinfinite words to biinfinite sequences of infinite sequences of weights, which we shall call a \emph{bidiverging automaton}, and in \S\ref{sec:bidiverging-power-series} we shall likewise define a kind of power series over biinfinite words, which we shall call \emph{bidiverging power series}, as well as a set of rational operations on these power series, and then we shall prove a Kleene Theorem connecting rational bidiverging power series and bidiverging automata.  In \S\ref{sec:quantum-simulation} we shall discuss how bidiverging automata are applied in practice for the purpose of simulating quantum systems.  Finally, in \S\ref{sec:conclusions} we shall present conclusions.

\section{Diverging Automata}
\label{sec:diverging-automata}

\subsection{Preliminary Formalism}

In this section we shall define a form of automaton that has the behavior of mapping infinite words to infinite sequences of values from a semiring.  First, however, we need to add some structure to infinite sequences in order for us to be able to use them as coefficients on infinite words.

Let $S$ be an arbitrary semiring.  It will often prove convenient to take advantage of the fact that $S^\nats \cong \nats\to S$ in order to be able to use function notation for defining sequences, so if $v\in S^\nats$ then $v(i)$ refers to the (zero-based) $i^{\text{th}}$ element of $v$, and $i\mapsto v(i)$ defines a sequence that is equal to $v$.  We now define three operations on $S^\nats$:  addition, left-multiplication by members of $S$, and right-multiplication by members of $S$.  For all sequences $a,b\in S^\nats$, we define $a+b := (i\mapsto a(i)+b(i))$, i.e. addition acts pointwise on the elements of the sequences.  Given $s\in S$ and $v\in S^\nats$, we define left multiplication by letting $s\cdot v = sv := (i\mapsto sv(i))$, i.e. left-multiplication by elements of $S$ multiplies all elements in the sequence by that factor, and likewise we define right-multiplication by letting $v\cdot s = vs := (i\mapsto v(i)s)$.  Finally, we let the zero (additive identity) of $S^\nats$ be the zero sequence, $i\mapsto 0$.  It is not difficult to see that all of the semimodule laws are obeyed by these operations and the choice of zero, and so we conclude that $S^\nats$ is both a left and a right $S$-semimodule (that is, an $S$-\emph{semibimodule}).

Let $A$ be some alphabet.  Words from this alphabet include elements of $A^*$, which we call \emph{finite} words; elements from $A^\omega\cong A^\nats$, which we call \emph{infinite} words; and elements from $A^\zeta\cong A^\ints$, which we call \emph{biinfinite} words.  In all cases the character of a word at zero-based position $i$ is denoted by $\charat{w}{i}$, so that for example $\charat{w}{2}$ refers to the \emph{third} character of a finite or infinite word, and $\charat{w}{-3}$ refers to the character at position $-3$ of a biinfinite word.

The notation $\slice{w}{s}{e}$ denotes the substring of $w$ starting at the $s^\text{th}$ character and going up to but not including the $e^{\text{th}}$ character.  If $e$ is finite then $\slice{w}{s}{e}=\charat{w}{s}\charat{w}{s+1}\dots \charat{w}{e-1}$ and if $e$ is $\infty$ then $\slice{w}{s}{e}=\charat{w}{s}\charat{w}{s+1}\dots$;  in both cases $\charat{\paren{\slice{w}{s}{e}}}{i} := \charat{w}{s+i}$.  The length of a finite word $w$ is given by $|w|$; note that $|\slice{w}{s}{e}|=e-s$.

Words can be concatenated to form other words, so that for example given words $a\in A^*$ and $b\in A^\omega$ the word $ab$ is the sequence $a$ followed by the sequence $b$, i.e.
$$(ab)(i) = \begin{cases}
a(i) & i < |a| \\
b(i-|a|) & i \ge |a|
\end{cases}.
$$

Note that in the sequel we will continue to use $S$ to denote an arbitrary semiring and $A$ an arbitrary alphabet unless stated otherwise.

\subsection{Diverging Automata Defined}

We now define the structure of the automaton as follows:  Given an alphabet $A$ and a semiring $S$, we define a diverging automaton $\aut{A}$ over $A$ and $S$ to be a tuple, $\auttuple{A}$, where
\begin{itemize}
\item $\ofaut{Q}{A}$ is the set of states in the automata;
\item $\ofaut{I}{A}\in S^Q$ is the initial distribution of states;
\item $\ofaut{F}{A}\in S^Q$ is the final distribution of states; and
\item $\ofaut{M}{A}\in S^{A\times Q\times Q}\cong (\PSA)^{Q\times Q}$ is the tensor providing the weighted transitions between states for each input symbol --- or, equivalently, a matrix of formal power series over the input alphabet $A$ with coefficients in $S$.
\end{itemize}

It will be convenient to establish some conventions at this point.  First, when defining an automaton $\aut{A}$ we shall do so by defining $Q^{\aut{A}}$, $I^{\aut{A}}$, etc. with the understanding that $\aut{A}:=\auttuple{A}$ (where the superscript on the elements of the tuple denote the automaton with which they are associated) so we do not need to relate the elements of the tuple to $\aut{A}$ explicitly.  Second, when we refer to the states in $\ofaut{I}{A}$ or the states in $\ofaut{F}{A}$ we shall be implicitly referring to the states \emph{with non-zero weight}, which we shall call respectively the \emph{initial states} and the \emph{final states}, or the \emph{initial state} and the \emph{final state} if there is only one such state.  Third, the notation $\ofaut{Q}{A}\cup\{q\}$ will always mean adding a fresh state to $\ofaut{Q}{A}$ unless otherwise stated (that is, if there was already a state named $q$ in $\ofaut{Q}{A}$ then it will be implicitly be relabeled to $q'$ and the rest of the automaton adjusted accordingly).   Finally, if $R\subset \ofaut{Q}{A}$, then $\ofaut{I}{A}_R$ will refer to the block of $\ofaut{I}{A}$ for the states in $R$, and likewise for $\ofaut{F}{A}$ and (both subscripts of) $\ofaut{M}{A}$.

Note that at this point there is nothing that distinguishes this automaton from a weighted automaton over finite words which we will henceforth call a \emph{converging}\footnote{The word `converging' was chosen both to provide a nice contrast to `diverging' and also because the word `finite' has essentially already been taken in the context of automata to refer to the number of states --- as in `finite state automata'.} automaton.  The isomorphism between diverging and converging automata shall prove useful throughout this paper, and so we shall denote by $\tilde{\aut{A}}$ the \emph{converging counterpart} of $\aut{A}$, by which we mean the automaton that has the same tuple as $\aut{A}$.

$\ofaut{M}{A}$ induces the structure of a directed graph where each (directed) edge is labeled by a value in $\PSA$, i.e. by a sum over allowed input symbols with a coefficient from $S$ on each;  by convention, when we refer to the directed edges in the graph without a qualifier we will only be referring to the edges with a non-zero label, since these are usually the only ones we care about.

Because of the graph structure of an automaton, we can (and shall) meaningfully talk about the path or paths taken by a word, which we define as follows:  a \emph{path} for a word $w$ is a sequence $\{q_i\}_{0\le i \le |w|}$ (which may be infinite) such that for every $0 \le i < |w|$ we have that $\Mofautwithsup{A}{\charat{w}{i}}_{q_iq_{i+1}}\ne 0$.  A finite path is said to be \emph{successful} if it starts on an initial state and ends on a final state.

We now need to define the \emph{behavior} of a diverging automaton $\aut{A}$, denoted by $\behav{\aut{A}}\in \PSAdiverging$, which describes the formal power series recognized by $\aut{A}$ as a function of the elements in its tuple. 
For convenience, we observe that $\PSAdiverging\cong(A^\omega\times\mathbb{N}\to S)$, and so we shall denote by $\aut{A}(w,\cdot)$ the value of the (infinite sequence) coefficient on the infinite word $w$ and by $\aut{A}(w,n)$ the value at the (zero-based) $n^{\text{th}}$ position of $\aut{A}(w,\cdot)$.  Equivalently we have that,
$$\behav{\aut{A}} = \sum_{w\in A^\omega} (n\mapsto \aut{A}(w,n)) \cdot w.$$

To define the behavior, we recall that we stated previously that an important motivation is to model the rate at which the infinite sum-of-products in an automaton diverge for infinite words.  Given this, an obvious definition for the behavior of $\aut{A}$ is to let $\aut{A}(w,n)$ be equal to the weight that $\slice{w}{0}{n}$ has in $\tilde{\aut{A}}$ (which recall is the converging automaton with the same tuple as $\aut{A}$), namely
$$\aut{A}(w,n) := \ofaut{I}{A} \cdot \paren{\,\,\prod_{i=0}^{n-1} \Mofautwithsup{A}{\charat{w}{i}}} \cdot \ofaut{F}{A}.$$

(Note:  The definition above will eventually be modified for reasons that will be shown later, but we start with it in this form for pedagogical reasons;  the final form will be presented in \S\ref{subsub:final-form}.)

\subsection{Examples}

To see the consequences of this definition,  we shall now consider some examples.  For the automata defined in these examples, we shall use the standard form of diagrammatic notation:
\begin{itemize}
\item states are denoted by circles;
\item the set of non-zero (weighted) transitions are denoted by arrows between states which have labels of the form `$x:y$' where $x\in A$ is the input symbol for the transition and $y\in S$ is the weight of the transition;
\item the set of initial states are denoted by arrows with an ending state but no starting state and labeled with the initial weight; and
\item the set of final states are denoted by arrows with an starting state but no ending state and labeled with the final weight.
\end{itemize}
For the sake of being explicit, values will also be given for each of the quantities in the 4-tuple for each automaton;  the transition matrix $M$ will be given as a matrix of pairs of input symbols and semiring coefficients.

We first consider the automaton $\aut{A}_1$ in Figure \ref{fig:automata-example-1}.  It is not difficult to see that for all words in the language $\alpha^*\beta\alpha^\omega$ we have that $$\aut{A}_1(\alpha^m\beta\alpha^\omega,n) =
\begin{cases}
\text{F} & n < m + 1\\
\text{T} & n \ge m+1,
\end{cases}
$$
because for a prefix of $m$ $\alpha$ symbols it takes at least $m+1$ steps to make it to the state $q_2$, which is the only final state, and once at $q_2$ the automaton loops forever (for this word).

\begin{figure}
	\centering
	\includegraphics[width=4.5in]{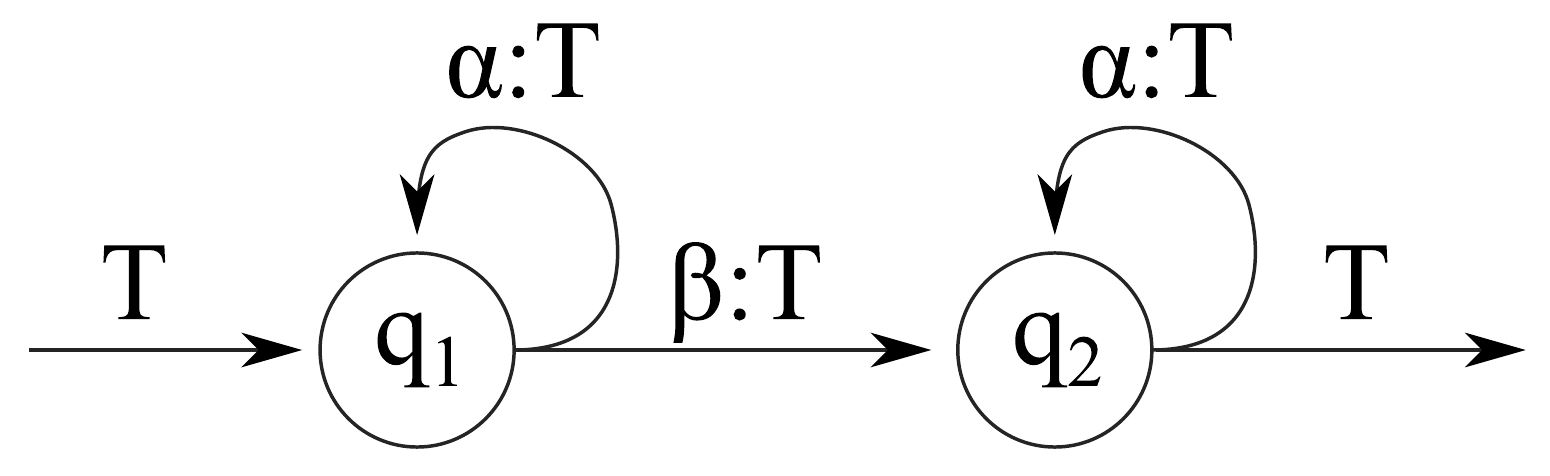}
	$A := \{\alpha,\beta\}$, $S := \mathbb{B}$, 
	$\aut{A}_1 :=
		\paren{
			\{q_1,q_2\},
			\begin{pmatrix}
			    \text{T} & \text{F}
			\end{pmatrix},
			\begin{pmatrix}
				\text{F} \\
				\text{T}
			\end{pmatrix},
			\begin{pmatrix}
				(\alpha,\text{T}) & (\beta,\text{T}) \\
				0 & (\alpha,\text{T})
			\end{pmatrix}
		}
	$
	\caption{Automaton Example \#1}
	\label{fig:automata-example-1}
\end{figure}
\begin{figure}
	\centering
	\includegraphics[width=4.5in]{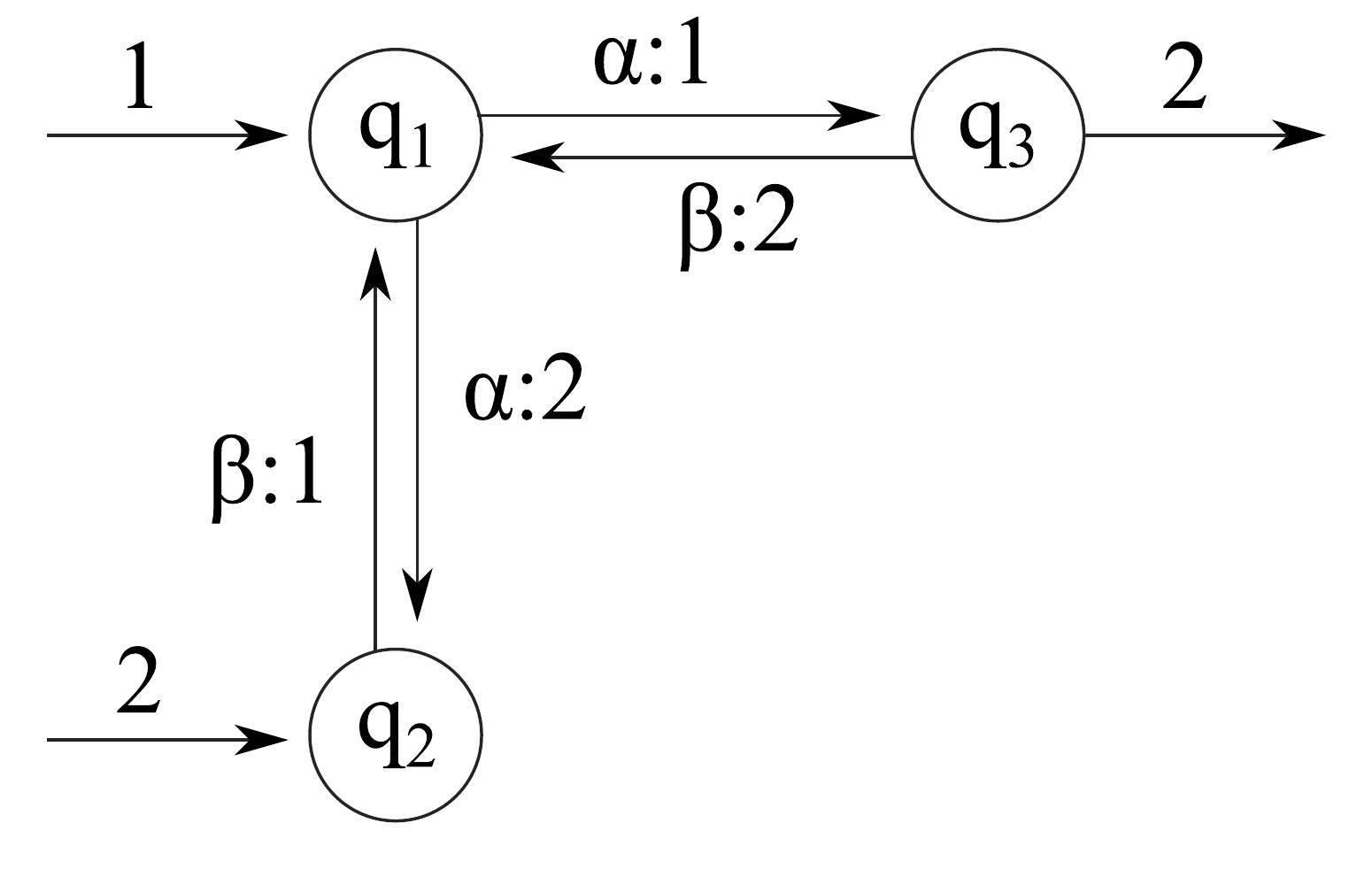}
	$A := \{\alpha,\beta\}$, $S := \mathbb{N}$, 
	$$\aut{A}_2 :=
		\paren{
			\{q_1,q_2,q_3\},
			\begin{pmatrix}
			    1 & 2 & 0
			\end{pmatrix},
			\begin{pmatrix}
				0 \\
				0 \\
				2
			\end{pmatrix},
			\begin{pmatrix}
				0 & (\alpha,2) & (\alpha,1) \\
				(\beta,1) & 0 & 0\\
				(\beta,2) & 0 & 0 \\
			\end{pmatrix}
		}
	$$
	\caption{Automaton Example \#2}
	\label{fig:automata-example-2}
\end{figure}
\begin{figure}
	\centering
	\includegraphics[width=4.5in]{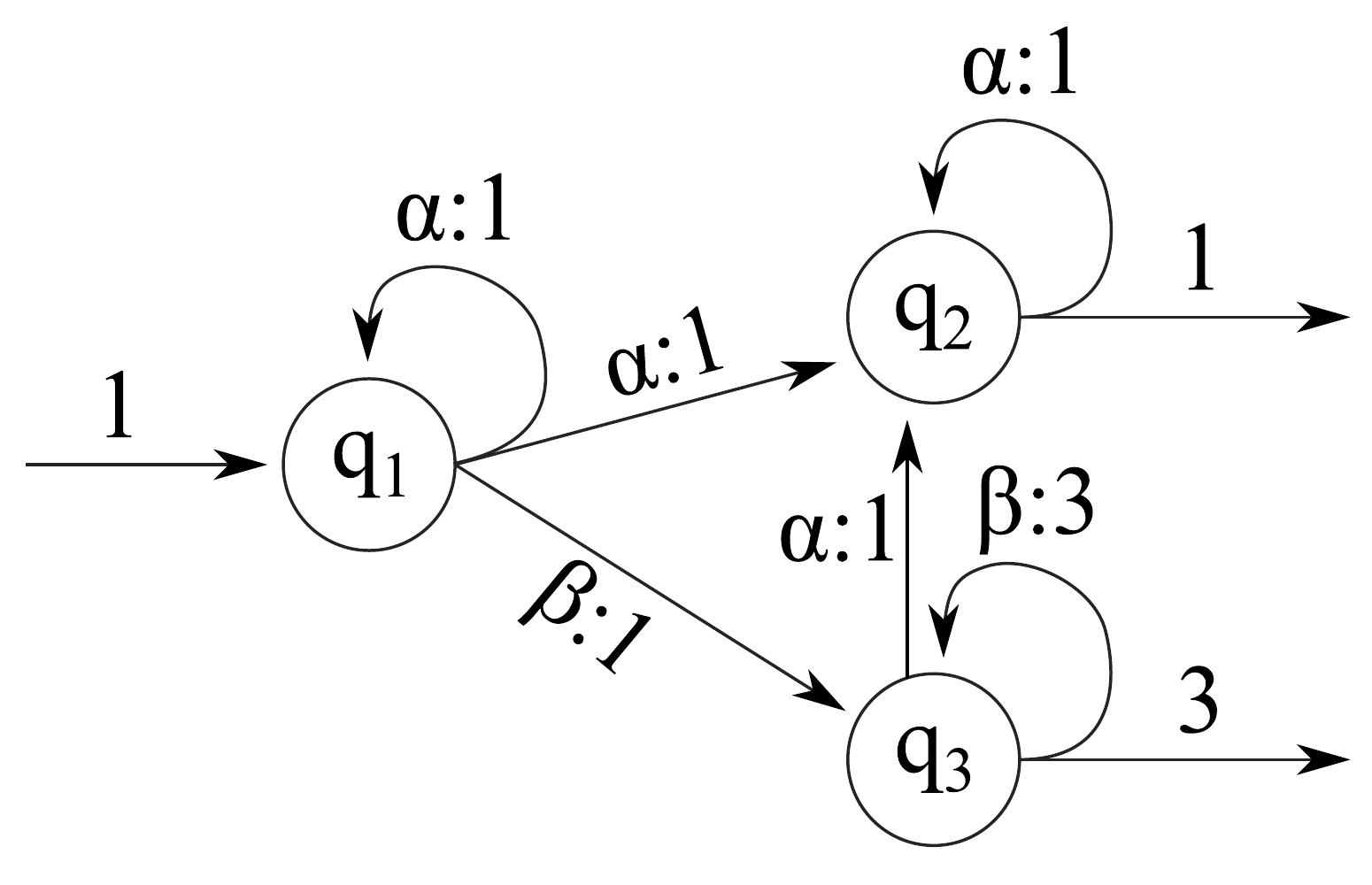}
	$A := \{\alpha,\beta\}$, $S := \mathbb{N}$, 
	$$\aut{A}_3 :=
		\paren{
			\{q_1,q_2,q_3\},
			\begin{pmatrix}
			    1 & 0 & 0
			\end{pmatrix},
			\begin{pmatrix}
				0 \\
				1 \\
				3 
			\end{pmatrix},
			\begin{pmatrix}
				(\alpha,1) & (\alpha,1) & (\beta,1) \\
				0 & (\alpha,1) & 0 \\
				0 & (\alpha,1) & (\beta,3) \\
			\end{pmatrix}
		}
	$$
	\caption{Automaton Example \#3}
	\label{fig:automata-example-3}
\end{figure}

We next consider the automaton $\aut{A}_2$ in Figure \ref{fig:automata-example-2}, where we observe that
$$\aut{A}_2\left((\alpha\beta)^\omega,n\right):=
\begin{cases}
0 & n \,\,\text{even}, \\
2^n & n \,\,\text{odd}.
\end{cases}
$$
To see why, first note that all successful paths must start at $q_1$ and end at $q_3$ because the former is the only state with outgoing transitions for $\alpha$ and the latter is the only final state.  Based on this, we immediately conclude that the value for $(\alpha\beta)^\omega$ must be zero for even $n$ because only an odd length path can make it from $q_1$ to $q_3$.  Next we note that the only successful paths with odd length $n$ are those which consist of $(n-1)/2$ round trips between $q_1$ and either $q_2$ or $q_3$, followed at the end by a step from $q_1$ to $q_3$ of weight 1.  There are $2^{(n-1)/2}$ independent paths of this form (as each round-trip in each sequence independently chooses $q_2$ or $q_3$ as a destination), and furthermore each path has total weight $2^{(n-1)/2}$ (a factor of 2 for each round trip), so because the final weight is 2, the value for $(\alpha\beta)^\omega$ for odd $n$ is therefore $2^{(n-1)/2}\cdot 2^{(n-1)/2}\cdot 2 =2^n$.  By applying similar reasoning, it is not hard to see that
$$\aut{A}_2\left((\beta\alpha)^\omega,n\right):=
\begin{cases}
0 & n = 0 \\
0 & n \,\,\text{odd}, \\
2^n & n \,\,\text{even}.
\end{cases}
$$
(If this result is not obvious, observe that $(\beta\alpha)^\omega = \beta(\alpha\beta)^\omega$ and that $q_2$ has initial weight 2.)

Finally, we consider the automaton $\aut{A}_3$ in Figure \ref{fig:automata-example-3}, for which we observe that
$\aut{A}_3(\alpha^\omega,n)$ $ = n$ due to the fact that for every $n$ there are are $n$ paths of length $n$ starting at $q_1$ and ending at $q_2$, and each path has weight 1.  It is also straightforward to see that
$$\aut{A}_3(\beta^\omega,n) =
\begin{cases}
0 & n = 0 \\
3^n & n > 0 \\
\end{cases}
$$ and $$\aut{A}_3(\beta^m\alpha^\omega,n) =
\begin{cases}
0 & n = 0 \\
3^n & 0 < n < m \\
3^m & n \ge m
\end{cases}
$$

\subsection{A Problem}

In the above examples we provided some illustrations for how the definition of the behavior works out in practice, but there is one very important property of this definition that has not yet been touched upon:   since the computation of $\aut{A}(w,n)$ only depends on the first $n$ characters of $w$, \emph{all} words with the same prefix of length $n$ have the same value at $n$ in their sequences --- that is, given words $w$ and $v$ such that $\slice{w}{0}{n}=\slice{v}{0}{n}$, it immediately follows that $\aut{A}(w,n)=\aut{A}(v,n)$.  This property is unfortunate  because it means that there are many strings that intuitively should be \emph{entirely rejected} by a given automaton --- that is, mapped to the zero element of $S^\nats$, which is the sequence with all zero entries --- that instead are \emph{accepted} by the automaton --- that is, mapped to a sequence with non-zero entries.

To see examples of this, first consider again the automaton in Figure \ref{fig:automata-example-1}.  This automaton is designed to filter out strings with more than a single $\beta$, and yet $\aut{A}_1(\beta\beta\alpha^\omega,n) = \delta_{n1} \ne 0$, where $\delta_{ij}$ is the Kronecker delta,
$$\delta_{ij} = \begin{cases}
1 & i = j, \\
0 & \otherwise.
\end{cases}$$
Likewise, if we consider again the automaton in Figure \ref{fig:automata-example-2} we see that although strings with more than a single $\alpha$ ought be rejected entirely we actually have that $\aut{A}_2(\alpha^\omega,n) = 2\delta_{n1} \ne 0$.
Finally if we consider again the automaton in Figure \ref{fig:automata-example-3} we see that although no $\beta$ should follow an $\alpha$, we actually have that $\aut{A}_3(\alpha\beta^\omega,n) = \delta_{n1}$.

Put another way, given an arbitrary language $L$ we would (naively) expect that if we took a B\"uchi automaton that recognized $L$ and converted it to a diverging automaton by labeling the existing transitions with weight 1 (in the semiring $\bools$) and the non-existing transitions with weight 0, then we would end up with an automaton that only accepted words in $L$, but from the preceding discussion we know that this will not be in true in general.

\subsection{A Fix}
\label{subsub:final-form}

The problems described above ultimately come from the fact that our definition for the behavior has the undesirable property that later parts of the word cannot affect early parts of the sequence, so by the time a word has hit a dead end that would have caused it to be rejected were it finite, it has already generated a non-zero subsequence and hence cannot be entirely rejected under the current definition of the behavior.  So, in a matter of speaking, if we want to make the behavior of our automata more sensible, we need a way for this future information to travel backwards in time to the beginning of the word.

Fortunately, we can do exactly this in a way that borrows a page from the B\"uchi playbook (see Refs. \cite{Buchi1960} and \cite{Buchi1960a}).  In a B\"uchi automaton at least one of the final states must be \emph{visited infinitely often} for an infinite word to be accepted.  This condition provides exactly what we need, because a word will hit a dead end if and only if it fails to visit a final state infinitely often, though because our automata have weights we need to be a bit more careful about how we define `visit infinitely often' because it is possible for there to be multiple paths that are individually non-zero but which cancel when they meet at particular states.  Furthermore, it makes sense to use a slightly weaker condition because all we really need is to ensure that there is no $n_0$ such that $\tilde{\aut{A}}(w,n)=0$ for all $n\ge n_0$.  Finally, it will turn out to be important that we also define our condition in terms of pairs of initial and final states rather than just final states so that we can express arbitrary automata as a union of automata with only a single initial and final state.  With these considerations in mind, given an initial state $i$ and a final state $f$ we shall say that a word $w$ \emph{activates} $(i,f)$ if for every $n_0\in\nats$ there exists $n\ge n_0$ such that the sum of all paths for $\slice{w}{0}{n}$ starting on $i$ and ending on $f$ is non-zero.  We shall call this rule the \emph{activation condition}, as it is not quite the same as an acceptance condition because it specifies not only whether a word is accepted but which initial and final states will be used when calculating the value of $\aut{A}(w,n)$.

Employing the above activation condition, we modify our definition of the behavior of a diverging automaton as follows:
$$\aut{A}(w,n) := \ofaut{I}{A} \cdot V^{\aut{A}}\paren{w,\,\,\,\,\prod_{i=0}^{n-1} \Mofautwithsup{A}{\charat{w}{i}}}\cdot \ofaut{F}{A},$$
where $V^{\aut{A}}(w,x)_{ij}=x_{ij}$ if $(i,j)$ has been activated by $w$ and $V^{\aut{A}}(w,x)_{ij}=0$ otherwise.  

Similar to the case of finite words, the right-hand side can equivalently be interpreted as the sum of all the successful paths taken by $\slice{w}{0}{n}$ where each path bears a weight equal to the product of transitions along the path ($\prod_{i=0}^{n-1} \Mofautwithsup{A}{\charat{w}{i}}$), the weight of the initial state of the path ($\ofaut{I}{A}$), the weight of the final state of the path ($\ofaut{F}{A}$), and the extra condition imposed by $\ofaut{V}{A}$.  This follows from the fact that the above sum can alternatively be expressed as
$$\aut{A}(w,n) := \sum_{q_0,\dots,q_n\in Q}\ofaut{I}{A}_{q_0} \Mofautwithsup{A}{\charat{w}{0}}_{q_0q_1} \Mofautwithsup{A}{\charat{w}{1}}_{q_1q_2}\cdots \Mofautwithsup{A}{\charat{w}{n-1}}_{q_{n-1}q_n} \ofaut{F}{A}_{q_n}\cdot v(w,q_0,q_n),$$
where $v(w,q_0,q_n)$ is 1 if $(q_0,q_n)$ have been activated by $w$ and 0 otherwise;  note that each non-zero term in the sum has a separate assignment of $\{q_i\}_{0\le i \le n}$ $\subseteq Q$ that corresponds to a successful path with a sequence of states equal to $\{q_i\}_{0\le i \le n}$ $\subset Q$, and vice versa.

With this new definition for the behavior, all of the problems that we listed earlier disappear because in each case no pairs of states become activated and so $V^{\aut{A}}(w,x)=0$.  However, in cases where all states become activated, converging and diverging automata behave similar to how our old definition played out, as shown in the following Lemma.

\begin{lemma}[Conditions under which diverging matches converging]
\label{lem:diverging-behaves-same-way-as-converging-counterpart}
Given a diverging automaton $\aut{A}$ and its converging counterpart $\tilde{\aut{A}}$, if for some infinite word $w$ we have that all pairs of initial and final states have been activated then for all $n\in\nats$ we have that $\aut{A}(w,n)=\tilde{\aut{A}}(\slice{w}{0}{n})$.
\end{lemma}

\begin{proof}
Since all initial and final states have been activated, we have that $V^{\aut{A}}(w,x)_{ij}=x_{ij}$ for every matrix $x$, initial state $i$, and final state $j$, and therefore $V^{\aut{A}}(w,x)_{ij}=0$ only when $\ofaut{I}{A}_i=0$ or $\ofaut{F}{A}_j=0$.  We thus see that we can replace $V^{\aut{A}}(w,x)$ with $x$ and obtain the same result.  The remainder of the proof follows immediately from the definitions of the behaviors of diverging and converging automata and the fact that the tuples of $\aut{A}$ and $\tilde{\aut{A}}$ are equal.
\end{proof}

\subsection{Elementary Operations}

Before leaving the subject of diverging automata, we take a moment to define a couple of elementary operations.  First, given an automaton $\aut{A}$ (converging or diverging) over some semiring $S$ and scalar values $l,r\in S$, we define $l\cdot\aut{A}\cdot r =l\aut{A} r$ such that $\defsameofaut{Q}{\mathit{l}A\mathit{r}}{A}$, $\ofaut{I}{\mathit{l}A\mathit{r}}_i := l\cdot\ofaut{I}{A}_i$, $\ofaut{F}{\mathit{l}A\mathit{r}}_i:=\ofaut{F}{A}_i\cdot r$, and $\defsameofaut{M}{\mathit{l}A\mathit{r}}{A}$ --- that is we left-multiply the initial state vector by $l$ and we right-multiply the final state vector by $r$ and we leave everything else as is.  Second, given another automaton $\aut{B}$, we have that $\aut{A}+\aut{B}$ is given by
\begin{align*}
\ofaut{Q}{A+B} &:=\ofaut{Q}{A}+\ofaut{Q}{B} \\
\ofaut{I}{A+B}_i &:=
    \begin{cases}
        \ofaut{I}{A}_i & i\in \ofaut{Q}{A} \\
        \ofaut{I}{B}_i & i\in \ofaut{Q}{B}
    \end{cases} \\
\ofaut{F}{A+B}_i &:=
    \begin{cases}
        \ofaut{F}{A}_i & i\in \ofaut{Q}{A} \\
        \ofaut{F}{B}_i & i\in \ofaut{Q}{B}
    \end{cases} \\
\ofaut{M}{A+B}_{ij} &:=
    \begin{cases}
        \ofaut{M}{A}_{ij} & i,j\in \ofaut{Q}{A} \\
        \ofaut{M}{B}_{ij} & i,j\in \ofaut{Q}{B} \\
        0 & \otherwise,
    \end{cases}
\end{align*}
that is, essentially the two automata are simply merged into a single automaton but kept separate from each other.  Finally, we define the zero (additive identity) automaton by letting $\ofaut{Q}{\mathrm{0}}:=\emptyset$, $\ofaut{I}{\mathrm{0}}_i := \ofaut{F}{\mathrm{0}}_i := 0$, and $\ofaut{M}{\mathrm{0}}_{ij}:=0$.  These operations make automata into $S$-semibimodules, and it turns out that the behavior operator $\behav{\cdot}$ is a semibimodule homomorphism, as the following Lemma demonstrates.

\begin{lemma}[Behavior is homomorphism (for converging and diverging automata)]
\label{lem:behavior-is-linear}
Given automata $\aut{A}$ and $\aut{B}$  over some semiring $S$ and scalar values $\alpha,\beta,\gamma,\delta\in S$ we have that $\behav{\alpha\aut{A}\gamma + \beta\aut{B}\delta} = \alpha\behav{\aut{A}}\gamma + \beta\behav{\aut{B}}\delta$
\end{lemma}

\begin{proof}
Follows straightforwardly from the definitions, so the proof has been left as an exercise for the reader.
\end{proof}

\section{Diverging Power Series}
\label{sec:diverging-power-series}

\subsection{Diverging Power Series Defined}

Having introduced a new form of automaton in the previous section, we shall now introduce a corresponding new kind of power series, and then prove a Kleene Theorem to formally connect the two constructions.  We start with some definitions.  Recall that $S$ is understood to be an arbitrary semiring and $A$ an arbitrary alphabet.  We then define the \emph{converging} power series to be the set of power series over $A^*$ with coefficients in $S$, which is denoted by $\PSAfinite$, and the \emph{diverging} power series to be the set of power series over $A^\omega$ with coefficients in $S^\nats$, which is denoted by $\PSAinfinite$.

For convenience we shall use function notation for power series, as we have for automata, so if $x\in\PSAfinite$ is a converging power series then $x(w)$ is the coefficient on the word $w$, and if $y\in\PSAinfinite$ is a diverging power series then $y(w,\cdot)$ is the coefficient on the word $w$, and $y(w,n)$ is the $n^{\text{th}}$ coefficient of $y(w,\cdot)$.

We will often be taking sums over substrings of words, so given a word $w\in A^*$, we will use the notation,
$$\sum_{s_1s_2\dots s_N=w} f(s_1,s_2,\dots,s_N,w)$$
to mean the sum of the value of $f$ over all strings $s_1,\dots,s_N$ such that the concatenation of $s_1$ through $s_N$ is equal to $w$.

Having established some basic notation, we move on to endowing $\PSAfinite$ with the standard $\,^*$-semiring structure.  Specifically, given the converging power series $x,y\in \PSAfinite$ we define addition by $x+y := w \mapsto x(w) + y(w)$; multiplication by $x\cdot y := w \mapsto \sum_{ab=w} x(a) y(b)$; the additive identity by $w\mapsto 0$; and the $\,^*$ operator by $x^* := \sum_{i=0}^\infty x^i$, the last of which is well-defined if and only if $x$ is \emph{proper} --- that is, $x(\emptystring)=0$, or equivalently $x\in\PSAfinitenoempty$.  These definitions make $\PSAfinite$ into a $\,^*$-semiring.

We can further make $\PSAfinite$ into a $S$-semibimodule as follows.  Let $s\in S$ and $v\in\PSAfinite$; then $s\cdot v=w\mapsto s\cdot v(w)$ and $v\cdot s=w\mapsto v(w)\cdot s.$  Showing that the semibimodules laws are obeyed is left as an exercise for the reader.

We now turn our attention to $\PSAinfinite$, which we shall also make a $S$-semibimodule as follows.  First, we shall define addition in the obvious way: given diverging power series $x,y\in \PSAinfinite$ we let $x+y := (w,n) \mapsto x(w,n) + y(w,n),$ and the zero element be $(w,n) \mapsto 0$.  Now let $s\in S$ and let $v\in\PSAinfinite$;  then $s\cdot v=(w,n)\mapsto s\cdot v(w,n)$ and $v\cdot s=(w,n)\mapsto v(w,n)\cdot s.$  Again, showing that the semibimodule laws are obeyed is left as an exercise for the reader.

\subsection{Rational Diverging Power Series}

It would be nice if we could proceed by making $(\PSAfinite,\PSAinfinite)$ form a semiring-semimodule pair, as there is a natural way to define left-multiplication, but unfortunately it turns out to be difficult to do this in a nice way while also incorporating the activation condition into our formalism of diverging power series.  Thus, instead we shall define two ways to construct elements of $\PSAinfinite$ using elements of $\PSAfinite$.  The first is \emph{infinite iteration}, denoted by $^\omega$, which is defined as follows.  Let $s$ be a proper \emph{converging} power series.  Then
$$s^\omega(w,n) := s^*(\slice{w}{0}{n})\cdot\allprefixes{w}{s^*},$$
where $\allprefixes{w}{x}=1$ if for every $n_0$ there exists $n\ge n_0$ such that $x(\slice{w}{0}{n})\ne 0$ and $\allprefixes{w}{x}=0$ otherwise.

The second way to build a diverging power series is \emph{conjoining}, denoted by $\cdot\star\cdot$, which is defined as follows.  Let $x,y\in\PSAfinitenoempty$;  then the \emph{conjoin} of $x$ and $y$ is given by,
$$(x\star y)(w,n)=(xy^*)(\slice{w}{0}{n})\cdot\allprefixes{w}{xy^*}.$$

Having defined these two ways of building diverging power series from converging power series, we shall now define the rational diverging power series, denoted by $\Rat^\omega(S,A)$.  First, for convenience, let $\Rat^*(S,A)\subset\PSAfinite$ be the set of rational converging power series, and $\ratnoempty$ be the set of \emph{proper} rational converging power series.  We then define $\Rat^\omega(S,A)$ to be the smallest subset of $\PSAinfinite$ such that
\begin{enumerate}
\item $\Rat^\omega(S,A)$ is closed under finite sums;
\item $\Rat^\omega(S,A)$ is closed under left and right multiplication by elements of $S$;
\item for every $x,y\in \ratnoempty$, $x\star y\in\Rat^\omega(S,A)$.
\item for every $z\in \ratnoempty$, $z^\omega\in\Rat^\omega(S,A)$;
\end{enumerate}

The next Lemma shows that there is a simple characterization of $\Rat^\omega(S,A)$.

\begin{lemma}[Characteristic representation for diverging power series]
\label{lem:characteristic}
A diverging power series $p\in \PSAinfinite$ is rational if and only if there exist finite index sets $I$ and $J$ and indexed sequences $\{a_i,b_i\}_{i\in I}\subset S$, $\{x_i,y_i\}_{i\in I}\subset \Rat^*_{/\emptystring}(S,A)$, $\{c_j,d_j\}_{j\in J}\subset S$ and $\{z_j\}_{j\in J}\subset \Rat^*_{/\emptystring}(S,A)$ such that
$$p=\sum_{i\in I}a_i(x_i\star y_i)b_i + \sum_{j\in J} c_jz_j^\omega d_j$$.
\end{lemma}

(Both the statement of this Lemma and its proof have well-known analogues in the case of unweighted words;  see Theorem 3.2 in Chapter 1 of Ref. \cite{Perrin2004}.)

\begin{proof}
Given the $I$, $J$, $\{a_i,b_i\}_{i\in I}$, $\{x_i,y_i\}_{i\in I}$, $\{c_j,d_j\}_{j\in J}$, $\{z_j\}_{j\in J}$ described in this Lemma, it is easy to see that $\sum_{i\in I}a_i(x_i\star y_i)b_i + \sum_{j\in J} c_jz_j^\omega d_j$ is rational, so the details are left as an exercise for the reader.

Now let $X$ be the set of diverging power series which can be written in the form $\sum_{i\in I}a_i(x_i\star y_i)b_i + \sum_{j\in J} l_jz_j^\omega r_j$.  Observe that:
\begin{enumerate}
\item $X$ is closed under finite sums.
\item $X$ is closed under left and right multiplication by elements of $S$ because each term has a coefficient on the left and right that can absorb values multiplied respectively on the left and right, and because $\PSAinfinite$ is an $S$-semibimodule (and thus distributive) we therefore have that $lvr\in X$ for all $l,r\in S$ and $v\in X$.
\item For all $z\in\Rat^*_{/\emptystring}(S,A)$, $z^\omega\in X$.
\item For all $x,y\in\Rat^*_{/\emptystring}(S,A)$, $x\star y\in X$.
\end{enumerate}
The set $X$ therefore contains $\Rat^\omega(S,A)$, and so we are done.
\end{proof}

Having defined rational power series, we now define recognizable power series as follows:  The set of recognizable diverging power series over the semiring $S$ and the alphabet $A$, $\Rec^\omega (S,A)\subset\powerseries{S^\nats}{A^\omega}$, is exactly the set of diverging power series that are the behavior of some diverging automaton, i.e. the set such that for every $x\in\Rec^\omega(S,A)$ there exists a diverging automaton $\aut{A}$ such that $\behav{\aut{A}}=x$; analogously, we define the set of recognizable converging power series, denoted by $\Rec^*(S,A)$, to be the set of power series that are the behavior of some converging automaton.

With this terminology we shall now state a Kleene Theorem that connects rational and recognizable series.

\begin{theorem}[Kleene's Theorem for diverging power series]
\label{thm:kleene-diverging}
$\Rat^\omega(S,A)=\Rec^\omega(S,A).$
\end{theorem}

Proving this result will take up the remainder of this section.

\subsection{Known Results about Finite Power Series}

We shall build our way to the proof of this Theorem by first proving a series of intermediate results.  Lemma \ref{lem:characteristic} tells us that we can decompose diverging power series into a finite number of operations on converging power series, and conversely we can build any diverging power series using a finite number of operations on converging power series.  Thus, if we had a way to immediately translate a converging power series to converging automata and back then we would be well on our way to proving the main theorem;  fortunately this is exactly what we have in the form of the very well-known Kleene-Sch\"utzenberger Theorem:

\begin{theorem}[Kleene's Theorem for converging power series]
\label{thm:kleene-converging}
$\Rat^*(S,A)=\Rec^*(S,A).$
\end{theorem}

\begin{proof}
See Refs. \cite{Shutzenberger1961}, \cite{Sakarovitch1987}, \cite{Kuich1997} and \cite{Esik2009}.
\end{proof}

We will want to connect automata together in various configurations, so to make this easy it would be nice if could express an arbitrary converging automaton in a form that only has a single starting and ending point.  Specifically, we prefer to work with automata which we shall call \emph{normalized} automata, which have exactly one initial state, which has no incoming edges, and exactly one (separate) final state, which has no outgoing edges, with both states having their respective initial and final weight equal to 1.  That is, if $\aut{A}$ is normalized, $1$ is the initial state, and $2$ is the final state, then we have that $\ofaut{I}{A}_i=\delta_{i1}$, $\ofaut{F}{A}_i=\delta_{i2}$, and $\ofaut{M}{A}_{i1}=\ofaut{M}{A}_{2i}=0$.  These automata have the property that they must reject the empty word, as the following Lemma shows.

\begin{lemma}
\label{lem:normalized-rejects-empty-string}
If $\aut{N}$ is a normalized automaton then it rejects the empty word.
\end{lemma}

\begin{proof}
This comes directly from the fact that the initial states and the final states have no overlap.
\end{proof}

We fortunately have a well-known result that tells us we can always assume we are working with a normalized automaton --- though we shall specifically be proving a variation of this result that assumes that the automaton does not recognize the empty word because we do not have $\epsilon$-transitions in our definition of automata.

\begin{lemma}[Qualified existence of an equivalent normalized automaton]
\label{lem:normalized}
For every converging automaton $\aut{A}$ that rejects the empty word there exists a normalized converging automaton $\aut{N}$ that has the same behavior as $\aut{A}$.
\end{lemma}

\begin{proof}
Let $\aut{N}$ be defined as follows:  Let $\ofaut{Q}{N} := \ofaut{Q}{A}\cup \{1,2\}$, $\ofaut{I}{N}_i := \delta_{i1}$, $\ofaut{F}{N}_i=\delta_{i2}$, and
$$\ofaut{M}{N}_{ij} :=
\begin{cases}
\ofaut{M}{A}_{ij} & i,j, \in \ofaut{Q}{A} \\
\sum_{k\in Q} \ofaut{I}{A}_k\ofaut{M}{A}_{kj} & i = 1, j\in\ofaut{Q}{A} \\
\sum_{k\in Q} \ofaut{M}{A}_{ik}\ofaut{F}{A}_k & i\in\ofaut{Q}{A}, j = 2 \\
\sum_{i,j\in Q} \ofaut{I}{A}_i \ofaut{M}{A}_{ij}\ofaut{F}{A}_j & i = 1, j = 2 \\
\end{cases}
$$
What we have done is create a new initial state and a new final state, and to the former added copies of all of the edges outgoing from the initial states in $\aut{A}$, multiplying the weights on these edges by the initial weight of that state, and to the latter added copies of all the edges incoming to the final states in $\aut{A}$, multiplying the weights on these edges by the final weight of the state.

Now let $w$ be an arbitrary word.  First observe that if $|w|=0$ then $\aut{N}(w)=0=\aut{A}(w)$ because there is no overlap between the initial and final states.  Next observe that if $|w|=1$ and $\charat{w}{0}=a$ for arbitrary $a\in A$ then
$$\aut{N}(a) = \ofaut{I}{N} \cdot \Mofautwithsup{N}{a} \cdot \ofaut{F}{N} = \sum_{i,j\in Q}  \ofaut{I}{A}_i  \ofaut{M}{A}_{ij} \ofaut{F}{A}_j = \ofaut{I}{A}\cdot \Mofautwithsup{A}{a} \cdot \ofaut{F}{A} = \aut{A}(a).$$   Finally observe that if $|w| = n > 1$ then
\begin{align*}
\aut{N}(w)
    &= \ofaut{I}{N} \cdot \Mofautwithsup{N}{\charat{w}{0}} \cdots \Mofautwithsup{N}{\charat{w}{n-1}} \cdot \ofaut{F}{N} \\
    &= \sum_{i,j\in Q} \ofaut{I}{A}_i \cdot \Mofautwithsup{A}{\charat{w}{0}}_{i,Q} \cdots \Mofautwithsup{A}{\charat{w}{n-1}}_{Q,j} \cdot \ofaut{F}{A}_j \\
    &= \aut{A}(w),
\end{align*}
and we are done.
\end{proof}

\begin{corollary}
\label{cor:kleen-normalized}
For every proper rational converging power series, there exists a normalized converging automaton that recognizes it.
\end{corollary}

\begin{proof}
Follows immediately from Theorem \ref{thm:kleene-converging} and Lemma \ref{lem:normalized}.
\end{proof}

\subsection{Loopback Automata}

It will also be useful to work with normalized automata with the property that the initial and final state are the same, so we shall define a \emph{loopback} automaton to be an automaton with the property that $I_i=F_i=\delta_{1i}$ --- that is, such that there is only a single state, with initial and final weight one, that is the only initial and final state;  we shall call this state the \emph{loopback} state.

It will be useful to categorize the ways that paths travel through loopback automata, so we say that the number of times that a path has made a circuit returning to the loopback state is equal to the number of \emph{trips} it has made, so in particular a \emph{single-trip path} is a path that starts and ends on the loopback state but does not pass through it again in between.

There is a natural transformation called \emph{rolling} that takes us from a normalized automaton to a loopback automaton:  Given a normalized automaton $\aut{N}$ with initial state $1$ and final state $2$, let the \emph{roll} of $\aut{N}$ be the automaton $\aut{L}$ given by $\ofaut{Q}{L}:=\ofaut{Q}{N}/\{2\}$, $\ofaut{I}{L}_i:=\ofaut{F}{L}_i:=\delta_{i1}$, and
$$
\ofaut{M}{L}_{ij}:=
\begin{cases}
\ofaut{M}{N}_{i2} & j = 1 \\
\ofaut{M}{N}_{ij} & \otherwise
\end{cases}
$$
That is, we delete the final state, redirect all edges that ended on the final state to the initial state, and then set the final weight of the initial state to 1 so that the initial state is now the loopback state.

Rolling has an inverse operation called \emph{unrolling}:  Given a loopback automaton $\aut{L}$ with loopback state $1$, the \emph{unroll} of $\aut{L}$ is the normalized automaton $\aut{N}$, given by $\ofaut{Q}{N}:=\ofaut{Q}{L}\cup\{2\}$, $\ofaut{I}{N}_i:=\delta_{i1}$, $\ofaut{F}{N}_i:=\delta_{i2}$,
$$
\ofaut{M}{N}_{ij}:=
\begin{cases}
0 & j = 1 \,\,\text{or} \,\,i = 2 \\
\ofaut{M}{L}_{i1} & j = 2 \\
\ofaut{M}{L}_{ij} & \otherwise
\end{cases}
$$
That is, we add a new state with final weight 1, redirect all the edges ending on the loopback state so that they now end on the new state, and set the final weight of the loopback state to 0, with the end result that the old loopback state is now the initial state and the newly added state is the final state.

\begin{lemma}
\label{lem:rolling-and-unrolling-are-inverse-operations}
Rolling and unrolling are inverse operations (modulo possibly reordering states).
\end{lemma}

\begin{proof}
The only parts of the automaton impacted by these transformations are the initial state which changes to the loopback state and back again, and the final state which is deleted and re-added (and vice versa).  The initial and final weights of these states are fixed since the automaton is either normalized or loopback, and so the inverse will always restore them to their original values (modulo possibly reordering states).  The edges are left unchanged except for those that end either at the loopback state or at the final state;  because normalized automata have no edges ending at the initial state, rolling essentially just has the effect of interchanging a zero column in $M$ with a non-zero column and then deleting the (interchanged) zero column, which is exactly inverted by the unrolling operation, and vice versa.
\end{proof}

The following Proposition gives us a useful specialization of Kleene's Theorem for the case of converging loopback automata.

\begin{proposition}[Converging loopback automata recognize $\,^*$ of rational power series]
\label{prop:loopback-equals-rational-star}
The set of power series recognized by a converging loopback automaton is equal to $\{s^*:s\in\ratnoempty\}$.
\end{proposition}

There are a couple of preliminary results that will be useful for proving this Proposition.

\begin{lemma}[Sum of single-trip paths in a loopback automaton equals sum in the unroll]
\label{lem:unroll-is-single-circuit}
Given a loopback automaton $\aut{A}$ and a finite non-empty word $w$, the sum over all single-trip paths is equal to the sum over all successful paths for $w$ in the unroll of $\aut{A}$.
\end{lemma}

\begin{proof}
Left as an exercise for the reader.
\end{proof}

\begin{lemma}[Behavior of converging loopback automata]
\label{lem:loopback-star}
Let $\aut{A}$ be a converging loopback automaton and $s$ be the proper converging power series recognized by its unroll.  Then $\behav{\aut{A}}=s^*$.
\end{lemma}

\begin{proof}
Let $w\in A^*$.  If $|w|=0$ then $\aut{A}(w)=s^*(w)=1$, so assume that $|w|>0$.  The set of all paths taken by $w$ that start and end at the loopback state can be partitioned into subsets based on the number of trips that they take.  Pick one of these subsets of paths --- say, the one with the paths that take $N$ trips for arbitrary $0 \le N \le |w|$ --- and then observe that this subset can be further subdivided into subsubsets such that every path in the subsubset visits the loopback state at exactly the same times, which means that we can express the sum over this subsubset as a product of factors where each factor is a sum over single-trip paths.  By Lemma \ref{lem:unroll-is-single-circuit} we conclude that each of these factors is equal to the weight of the corresponding substring in $s$, and therefore
$$\aut{A}(w) = \delta_{0|w|} + \sum_{N=1}^{|w|}\,\, \sum_{v_1\dots v_N=w}\,\,\prod_{k=1}^N s(v_k) = s^*(w).$$
\end{proof}

Now we are ready to prove our Proposition relating power series recognized by loopback automata and the $\,^*$ of rational converging power series.

\begin{proof}[Proof of Proposition \ref{prop:loopback-equals-rational-star}]
First assume that we have a converging loopback automaton $\aut{A}$.  Let $\aut{A}'$ be the unroll of this automaton.  Applying Kleene's Theorem for converging automata (Theorem \ref{thm:kleene-converging}) to $\aut{A}'$, we conclude that there exists a rational converging power series $s$ that is recognized by $\aut{A}'$, and because $\aut{A}'$ is normalized we know from Lemma \ref{lem:normalized-rejects-empty-string} that $s$ is proper; applying Lemma \ref{lem:loopback-star} we conclude that $\behav{\aut{A}}=s^*$ where $s$ is rational.

Now assume that $s$ is a proper rational converging power series.  Applying Corollary \ref{cor:kleen-normalized} (Kleene's Theorem plus normalization) we see that there exists a normalized automaton $\aut{A}'$ that recognizes $s$.  Let $\aut{A}$ be the roll of $\aut{A}'$;  applying Lemma \ref{lem:loopback-star} we conclude that $\behav{\aut{A}}=s^*$, and so we are done.
\end{proof}

There is an analog of Proposition \ref{prop:loopback-equals-rational-star} for the $\,^\omega$ operation.

\begin{proposition}[Diverging loopback automata recognize the $\,^\omega$ of rational power series]
\label{prop:loopback-equals-rational-omega}
The set of diverging power series recognized by a diverging loopback automaton is equal to $\{s^\omega:s\in\ratnoempty\}$.
\end{proposition}

There again will be a preliminary result that will be useful for proving this Proposition.

\begin{lemma}[Behavior of diverging loopback automata]
\label{lem:loopback-star-to-omega-conversion}
Let $\aut{A}$ be a diverging loopback automaton such that its converging counterpart, $\tilde{\aut{A}}$, recognizes the power series $s^*$.  Then $\behav{\aut{A}}=s^\omega$.
\end{lemma}

\begin{proof}
Let $w$ be some infinite word.  There are two cases:
\begin{enumerate}
\item If $w$ activates the loopback state, then, because it is the sole initial and final state, by Lemma \ref{lem:diverging-behaves-same-way-as-converging-counterpart} we have that $\aut{A}(w,n)=\tilde{\aut{A}}(\slice{w}{0}{n})=s^*(\slice{w}{0}{n})$.    Furthermore, the fact that $w$ activates the loopback state implies by definition that for every $n_0\in\nats$ there exists $n\ge n_0$ such that $\aut{A}(w,n)=s^*(\slice{w}{0}{n})\ne 0$, which means that $\allprefixes{w}{s^*}=1$ and so $s^\omega(w,n)=s^*(\slice{w}{0}{n})=\aut{A}(w,n)$ for all $n$.
\item If $w$ does not activate the loopback state, then $V^{\aut{A}}(w,x)_{ij}=0$ for every matrix $x$, initial state $i$, and final state $j$, and so $\aut{A}(w,n)=0$.  Furthermore, the fact that $w$ does not activate the loopback state implies by definition that there exists some $n_0$ such that for all $n\ge 0$ the sum of all successful paths for $\slice{w}{0}{n}$ is zero, so for all $n\ge n_0$ we also have that $\tilde{A}(\slice{w}{0}{n})=s^*(\slice{w}{0}{n})=0$, and therefore $\allprefixes{w}{s^*}=0$, and so for all $n\in\nats$ we have that $\aut{A}(w,n)=s^\omega(w,n)=0$.
\end{enumerate}
Thus we have shown that for all $w$ and all $n$, $\aut{A}(w,n)=s^\omega(w,n)$ which directly implies that $\behav{\aut{A}}=s^\omega$, and we are done.
\end{proof}

Now we have what we need to prove Proposition \ref{prop:loopback-equals-rational-omega}, which we recall equates the behavior of diverging loopback automata and the $^\omega$ operation applied to rational converging power series.

\begin{proof}[Proof of Proposition \ref{prop:loopback-equals-rational-omega}]
First let $\aut{A}$ be a diverging loopback automaton and $\tilde{\aut{A}}$ be its converging counterpart.  By Proposition \ref{prop:loopback-equals-rational-star} we know that $\behav{\tilde{\aut{A}}}=s^*$ for some proper rational converging power series $s$.  Applying Lemma \ref{lem:loopback-star-to-omega-conversion} we conclude that $\behav{\aut{A}}=s^\omega$.

Now let $s$ be a proper rational converging power series.  By Proposition \ref{prop:loopback-equals-rational-star} there exists a converging loopback automaton $\tilde{\aut{A}}$ that recognizes $s^*$;  let $\aut{A}$ be the diverging counterpart of $\tilde{\aut{A}}$.  Then applying Lemma \ref{lem:loopback-star-to-omega-conversion} we conclude that $\behav{\aut{A}}=s^\omega$.
\end{proof}

\subsection{Loopback Automata With Preludes}

There is another specialized kind of automaton that will prove useful:  We say that an automaton is \emph{loopback with prelude} if it has single initial state with no incoming edges and a (separate) single final/loopback state, both of with have weight 1;  we say that an automaton is a \emph{loopback without prelude} if it is an ordinary loopback automaton.  We say that an automaton is loopback \emph{with or without prelude} if it is either a loopback automaton or a loopback automaton with prelude.

One of the advantages of these categories is that we can express any automaton in terms of a weighted sum of them, as the following Lemma shows.

\begin{lemma}[Decomposition into loopback automata with or without prelude]
\label{lem:decompose-into-loopback-with-or-without-prelude}
For all automata $\aut{A}$ there exists a decomposition into a weighted sum of automata that are all loopback with or without prelude, i.e. a tuple $(K,$ $\{l_k,r_k\}_{k\in K},$ $\{A_k\}_{k\in K})$ such that $\behav{\aut{A}}=\sum_{k\in K} l_k \behav{\aut{A}_k} r_k$ where $K$ is an index set, $\{l_k,r_k\}_{k\in K}$ is an indexed set of coefficients in the underlying semiring $S$, and $\{\aut{A}_k\}_{k\in K}$ is an indexed set of loopback automata each of which is with or without prelude.
\end{lemma}

\begin{proof}
Let $\aut{A}$ be an automaton.  For all states $p$ and $q$ let $\aut{A}_{pq}$ be defined as follows:
\begin{itemize}
\item 
If $p=q$, then $\defsameofaut{Q}{A_{\textit{pq}}}{A}$, $\ofaut{I}{A_{\textit{pq}}}_i:=\ofaut{F}{A_{\textit{pq}}}_i:=\delta_{iq}$, and $\defsameofaut{M}{A_{\textit{pq}}}{A}$.
\item
If $p\ne q$, then $\ofaut{Q}{A_{\textit{pq}}}:=\ofaut{Q}{A}\cup\{1\}$, $\ofaut{I}{A_{\textit{pq}}}_i:=\delta_{i1}$, $\ofaut{F}{A_{\textit{pq}}}_i:=\delta_{iq}$, and
$$\ofaut{M}{A_{\textit{pq}}}:=
\begin{cases}
\ofaut{M}{A}_{ij} & i,j\in\ofaut{Q}{A} \\
\ofaut{M}{A}_{pj} & j\in\ofaut{Q}{A}, i = 1
\end{cases}
$$
\end{itemize}
If $p=q$ then observe that $\aut{A}_{pq}$ is a loopback automaton without prelude, and if $p\ne q$ then observe that $\aut{A}_{pq}$ is a loopback automaton with prelude, and in both cases observe that $\aut{A}_{pq}$ recognizes the same power series as that recognized by $\aut{A}$ with its sole initial state set to $p$ and its sole final state to $q$, with both weights set to 1.  (Proving that the manipulations in the $p\ne q$ case preserved this property is left as an exercise for the reader.)  Finally, observe that because the definition of $\behav{\aut{A}}$ is homomorphic with respect to the elements of $\ofaut{I}{A}$ and $\ofaut{F}{A}$ (and the activation condition enforced by the function $V^{\aut{A}}$ does not break
this\footnote{Note, however, that it \emph{would} have broken the property of being homomorphic if the activation condition had been defined only in terms of the final state and not in terms of pairs of initial and final states;  to see why, consider the case of an automaton $\aut{A}$ with two initial states, $i_1$ and $i_2$, and one final state, $f$, such that $(i_1,f)$ was activated but $(i_2,f)$ was not. Then construct an automaton $\aut{B}$ by making two copies of $\aut{A}$ and making only $i_1$ be initial in the first and only $i_2$ be initial in the second.  Observe that if the activation condition only applied to the final state rather than to pairs of initial and final states then these two automata would not be equivalent because in automaton $\aut{A}$ the paths from $i_2$ to $f$ would have contributed to the sum (as the paths from $i_1$ would have been sufficient to activate $f$ for all initial states) whereas in $\aut{B}$ they would not have.}
because it ignores the actual values of $\ofaut{I}{A}$ and $\ofaut{F}{A}$) we therefore have that $\behav{\aut{A}}=\sum_{i,j\in Q} \ofaut{I}{A}_i \behav{\aut{A}_{ij}}\ofaut{F}{A}_j=\sum_{k\in K} l_k \behav{\aut{A}_k} r_k$ where $K=Q^2$, $l_{(i,j)} := \ofaut{I}{A}_i$, and $r_{(i,j)} := \ofaut{F}{A}_j$.
\end{proof}

There are two important operations that we shall now define that allow preludes to be attached to and detached from loopback automata.  Given normalized automata $\aut{X}$ and $\aut{Y}$, $\aut{X\star Y}$ is the \emph{conjoin} of $\aut{X}$ and $\aut{Y}$, defined as follows:  First, let $\aut{B}$ denote the roll of $\aut{Y}$.  Furthermore, let $1/2$ be the initial/final state of $\aut{X}$, and $3$ be the loopback state of $\aut{B}$.  Then
\begin{align*}
\ofaut{Q}{X\star Y} &:= Q^{\aut{X}}/\{2\}\cup Q^{\aut{B}}, \\
\ofaut{I}{X\star Y}_i &:= \delta_{i1}, \\
\ofaut{F}{X\star Y}_i &:=\delta_{i3}, \,\,\text{and} \\
\ofaut{M}{X\star Y}_{ij} &:=
\begin{cases}
\ofaut{M}{X}_{ij} & i,j \in \ofaut{Q}{X}/\{2\} \\
\ofaut{M}{X}_{i2} & i \in \ofaut{Q}{X}/\{2\}, j = 3 \\
\ofaut{M}{B}_{ij} & i,j \in \ofaut{Q}{B} \\
0 & \otherwise \\
\end{cases}.
\end{align*}
That is, we take the direct sum of $\aut{X}$ and the roll of $\aut{Y}$, merge the final state of $\aut{X}$ with the loopback state of the rolled $\aut{Y}$ (including all edges), and set the initial weight of the final/loopback state to zero.

Given a loopback with prelude automaton, $\aut{A}$, the \emph{disjoin} of $\aut{A}$ is defined to be the pair of automata $(\aut{X},\aut{Y})$ defined as follows:  Let 1 be the initial state of $\aut{A}$ and 2 be the final/loopback state.  Then $\aut{X}$ is the normalized automaton given by
\begin{align*}
\ofaut{Q}{X} &:=Q^{\aut{A}}\\
\ofaut{I}{X}_i &:= \ofaut{I}{A}_i =\delta_{i1}\\
\ofaut{F}{X}_i &:= \ofaut{F}{A}_i =\delta_{i2}\\
\ofaut{M}{X}_{ij}&:=
\begin{cases}
0 & i = 2 \\
\ofaut{M}{A}_{ij} & \otherwise
\end{cases},
\end{align*}
that is, $\aut{X}$ is the result of deleting all edges that start on the final/loopback state of $\aut{A}$, and $\aut{Y}$ is the unroll of $\aut{B}$, which is given by $\ofaut{Q}{B} := \ofaut{Q}{A}$, $\ofaut{I}{B}_i := \ofaut{F}{B} :=\delta_{i2}$, and $\ofaut{M}{B}_{ij} :=\ofaut{M}{A}_{ij}$ --- that is, $\aut{B}$ is the result of making the final/loopback state also be the sole initial state.

\begin{lemma}[Conjoin is the inverse of disjoin]
\label{lem:attaching-and-detaching-are-inverses}
Given an automaton $\aut{A}$ with prelude, and letting $(\aut{X},\aut{Y})$ be the disjoin of $\aut{A}$, we have that $\behav{\aut{X}\star\aut{Y}}=\behav{A}$.
\end{lemma}

\begin{proof}
Let $\aut{Z} := \aut{X}\star\aut{Y}$.  First recall that by Lemma \ref{lem:rolling-and-unrolling-are-inverse-operations}, rolling and unrolling are inverse operations (modulo possibly reordering the states, which is irrelevant here).  Thus, we can let $\aut{B}$ be the roll of $\aut{Y}$ and perform our analysis in terms of $\aut{X}$ and $\aut{B}$.  Observe that the only states touched by disjoining and conjoining are the initial and final/loopback states.  Furthermore note that disjoining does not delete any states it so essentially creates two copies of $\aut{A}$ with the only difference being that the first copy $(\aut{X})$ deleted the edges outgoing from the final state and the second copy $(\aut{B})$ has the final state also be the initial state.  The act of conjoining takes a direct sum of $\aut{X}$ and $\aut{B}$, and merges the final state of $\aut{X}$ with the loopback state of $\aut{B}$, which effectively undoes the edge deletion in the construction of $\aut{X}$ in the sense that the same edges exist, although with ends in $\aut{B}$ instead of $\aut{X}$.  We thus see that the result of conjoining the disjoin of $\aut{A}$ is an automaton with two copies of $\aut{A}$, with the two separate final/loopback states merged into a single final/loopback state and all outgoing edges for this state in the first copy deleted.  Thus, every state in $\aut{Z}$ can be uniquely mapped into a state in $\aut{A}$ by erasing the information about which copy it came from, and this map has the property that every edge between two states in $\aut{Z}$ corresponds to an edge between the two corresponding states in $\aut{A}$ (though they may have come from different copies).  We thus see that every successful path in $\aut{Z}$ can be uniquely mapped to an equivalent path in $\aut{A}$.

The opposite is not necessarily true, however, as in principle for a particular successful path in $\aut{A}$ there could be several equivalent paths in $\aut{Z}$, where each successful path is distinguished by which copy of $\aut{A}$ it was in at a particular step.  Fortunately, we can eliminate this possibility by noting that all successful paths in $\aut{Z}$ must pass through the final/loopback state as this is the sole final state, and furthermore the very first time that a path lands on the final/loopback state it immediately and irreversibly moves from the states in the first copy to the states in the second copy as all of the outgoing edges for the final/loopback state end on states in the second copy and there are no other states that connect the two copies.  Thus, we have shown that every successful path in $\aut{A}$ is equivalent to exactly one successful path in $\aut{Z}$, and vice versa.  This is significant because it means that we can merge the two copies of $\aut{A}$ within $\aut{Z}$ --- i.e., by replacing each pair of equivalent states and their edges with a single state and set of edges, except for the final/loopback state which is already merged --- without changing its behavior, as for any word the successful paths will not be affected as the states and edges will be the same except for the fact that they will all be in a single copy of $\aut{A}$ rather than having an initial prelude take place in another copy of $\aut{A}$.  Because the merged automaton is exactly isomorphic to $\aut{A}$ we see that $\behav{\aut{A}}=\behav{\aut{Z}}$ and thus we are done.
\end{proof}

Loopback automata with preludes are useful because of the following fact:

\begin{lemma}[Behavior of conjoin is conjoin of behavior]
\label{lem:attach-to-form-xy}
Let $\aut{X}$ and $\aut{Y}$ be a normalized converging automaton.  Then $\behav{\aut{X}\star\aut{Y}}=x\star y$.
\end{lemma}

\begin{proof}
By Kleene's Theorem \ref{thm:kleene-converging} and Lemma \ref{lem:normalized-rejects-empty-string} we know that there exist $x,y\in\ratnoempty$ such that $\behav{\aut{X}}=x$ and $\behav{\aut{Y}}=y$.  Let $\aut{Z}:=\aut{X}\star\aut{Y}$.

Let $w$ be an arbitrary infinite word and $n$ a positive non-zero integer (as we know that the final weight of the initial state is 0 and hence $\aut{Z}(w,0)=0$).  Observe that all successful paths in $\aut{Z}$ must land on the final/loopback state at some point, and consider the set of paths for which this occurs for the first time at step $k$ of the path where $0<k\le n$.  We can factor the sum over these paths into the product of the sum over all length $k$ paths from the initial state to the final/loopback state, and the sum over all paths of length $n-k$ looping through the final/loopback state.  Because a successful path cannot access any of the states from $\aut{Y}$ until it has landed on the final/loopback state for the first time, and because all of the states in $\aut{X}$ are present in $\aut{Z}$ and accessible from the initial state, with the exception of the final state which has effectively been replaced by the final/loopback state which has the same incoming edges as the final state in $\aut{X}$, we see that the sum in the first of the two factors is exactly equivalent to the sum over all successful paths in $\aut{X}$ for the substring $\slice{w}{0}{k}$, which is equal to $x(\slice{w}{0}{k})$.  Using similar reasoning we conclude that the second of the two factors is equivalent to a sum over all successful paths of length $n-k$ for the substring $\slice{w}{k}{n}$ in the roll of $\aut{Y}$ (as the construction of $\aut{X}\star\aut{Y}$ rolls $\aut{Y}$ before merging it with $\aut{X}$), and therefore by Proposition \ref{prop:loopback-equals-rational-star} this sum is equal to $y^*(\slice{w}{k}{n})$.  Summing over $k$ (and recalling that $x(\emptystring)=0$) we see that for all $n\in\nats$ the sum over all successful paths for $\slice{w}{0}{n}$ is equal to $$\sum_{ab=\slice{w}{0}{n}} x(a)\cdot y^*(b)=(xy^*)(\slice{w}{0}{n}).$$  Thus, we see that the initial and final/loopback states will be activated if and only if for every $n_0\in\nats$ there exists $n\ge n_0$ such that $(xy^*)(\slice{w}{0}{n})\ne 0$, and therefore if and only if $\allprefixes{w}{xy^*}=1$.  We thus have that
$$\paren{\aut{X\star Y}}(w,n) = (xy^*)(\slice{w}{0}{n})\cdot\allprefixes{w}{xy^*}=(x\star y)(w,n)$$
and so $\behav{\aut{X}\star\aut{Y}}=x\star y$.
\end{proof}

\begin{lemma}[Behavior of diverging loopback automata with prelude]
\label{lem:loopback-with-prelude-equals-xyomega}
The set of diverging power series recognized by bidiverging loopback with prelude automata is equal to $\{x\star y: x,y\in\ratnoempty\}$.
\end{lemma}

\begin{proof}
First, let $x$ and $y$ be proper rational converging power series.  By Corollary \ref{cor:kleen-normalized} we know that there exist normalized automata $\aut{X}$ and $\aut{Y}$ that recognize respectively $x$ and $y$.  By Lemma \ref{lem:attach-to-form-xy} we know that conjoining these two automata forms a diverging loopback with prelude automaton that recognizes $x\star y$.

Now let $\aut{A}$ be a diverging automaton with prelude and let $(\aut{B},\aut{C})$ be the disjoin of $\aut{A}$.  By Lemma \ref{lem:attaching-and-detaching-are-inverses} we know that $\aut{B}\star\aut{C}$ is an automaton with the same behavior as $\aut{A}$, and furthermore by Lemma \ref{lem:attach-to-form-xy} we know that this behavior is equal to $b\star c$ where $\behav{\aut{B}}=b$ and $\behav{\aut{C}}=c$ are converging power series that we know are proper because of Lemma \ref{lem:normalized-rejects-empty-string} (as $\aut{B}$ and $\aut{C}$ are normalized).  Finally, by Theorem \ref{thm:kleene-converging} we know that $b$ and $c$ are rational.
\end{proof}

\subsection{Proof of the Kleene Theorem}

We now have everything that we need to prove our Kleene Theorem.

\begin{proof}[Proof of Theorem \ref{thm:kleene-diverging}]
First assume that we are given a rational diverging power series $p$.  By Lemma \ref{lem:characteristic} we know that there exists a finite index sets $I$ and $J$ and indexed sequences $\{a_i,b_i\}_{i\in I}\subset S$, $\{x_i,y_i\}_{i\in I}\subset \Rat^*_{/\emptystring}(S,A)$, $\{c_j,d_j\}_{j\in J}\subset S$ and $\{z_j\}_{j\in J}\subset \Rat^*_{/\emptystring}(S,A)$ such that
$$p=\sum_{i\in I}a_i(x_i\star y_i)b_i + \sum_{j\in J} c_jz_j^\omega d_j.$$
For each $i$ we know by Lemma \ref{lem:loopback-with-prelude-equals-xyomega} that there exists an diverging automaton $\aut{Y}_i$ that recognizes $x_i \star y_i$, and for each $j$ we know by Proposition \ref{prop:loopback-equals-rational-omega} that there exists a diverging automaton $\aut{Z}_j$ that recognizes $z_j^\omega$.  Let
$$\aut{P}:=\sum_{i\in I}a_i \aut{Y}_i b_i + \sum_{j\in J} c_j\aut{Z}_j^\omega d_j,$$
and by Lemma \ref{lem:behavior-is-linear} we have that $\behav{\aut{P}}=p$.

Now assume that we are given an automaton $\aut{P}$.  By Lemma \ref{lem:decompose-into-loopback-with-or-without-prelude} we know that $\behav{\aut{P}}$ can be expressed as a weighted sum of the behaviors of automata that are all loopback with or without prelude.  Since by Proposition \ref{prop:loopback-equals-rational-omega} we have that loopback automata recognize power series of the form $z^\omega$ with $z\in\ratnoempty$, and since by Lemma \ref{lem:loopback-with-prelude-equals-xyomega} we have that loopback automata with prelude recognize power series of the form $x\star y$ with $x,y\in\ratnoempty$, we conclude that the power series recognized by $\aut{P}$ is rational.
\end{proof}

\section{Bidiverging Automata}
\label{sec:bidiverging-automata}

\subsection{Preliminary Formalism}

In the previous sections we have presented automata and power series over the domain of infinite words.  These words were uni-infinite in the sense that they have a definite starting point and proceed towards infinity in a single direction.  When studying infinite systems in physics, however, we are usually interested in the case where there are no boundaries, which means that the system stretches out infinitely in all directions.  For this reason, in this and the next section we shall proceed to extend the formalism that has been developed so far into the domain of \emph{biinfinite} words.

Unlike diverging automata, bidiverging automata shall map biinfinite words to coefficients in $S^{\ints\times\nats}$, where the extra $\ints$ effectively adds an additional parameter that specifies the starting location in the word; this additional argument is needed because unlike the case of infinite words, in the case of biinfinite words there is not a natural location at which to start (and position 0 does not count because we can always shift the word left or right, making the location of position 0 itself an arbitrary choice).  As always, we observe that $S^{\ints\times\nats}\cong \ints\times\nats\to S$, which means that we can use function notation to describe and specify elements in $S^{\ints\times\nats}$.

We now endow $S^{\ints\times\nats}$ with the same kind of $S$-semibimodule structure with which we endowed $S^\nats$.  Specifically, given $x,y\in\bicoefs$, we define addition by $x+y := (i,n)\mapsto x(i,n) + y(i,n)$, given $s\in S$ we define left-multiplication by $s\cdot x = sx = (i,n)\mapsto s x(i,n)$ and right-multiplication by $x\cdot s = (i,n)\mapsto x(i,n) s$, and finally we define the additive identity to be $(i,n)\mapsto 0$.  It is not hard to see that these definitions obey the semibimodule laws and so $\bicoefs$ is an $S$-semibimodule.

Because biinfinite words extend in two directions, we need to extend our terminology in order to define boundary conditions for bidiverging automata.  Given an initial state $i$ and a final state $f$, we say that a biinfinite word $w$ \emph{activates} $(i,f)$ if for every $i_0\le j_0$ there exists $i\le i_0 \le j_0 \le j$ such that the sum of all successful paths for $\slice{w}{i}{j}$ is non-zero.

Note that this property is shift-invariant because if this property holds for one shift then it holds for any other shift as for any $i_0\le j_0$ we can shift the word back to where we know the property holds, adding or subtracting the size of the shift to $i_0\le j_0$ so that they follow the word, obtain $i$ and $j$ there, and then shift them back to where we started, and conversely if this not property does not hold for a particular shift of the word then it cannot hold for any other as, applying the previous argument, if the property did hold in one shift then it would hold for all shifts, leading to a contradiction.

\subsection{Bidiverging Automata Defined}

As with diverging automata we shall use function notation as a convenient means of defining the behavior, which we do as follows:
$$\aut{A}(w,i,n) := \ofaut{I}{A}\cdot V^{\aut{A}}\paren{w,\,\,\,\,\prod_{j=0}^{n-1} \Mofautwithsup{A}{\charat{w}{i+j}}}\cdot \ofaut{F}{A}$$
where $V^{\aut{A}}(w,x)_{ij}=x_{ij}$ if $w$ activates $(i,j)$ and $V^{\aut{A}}(w,x)_{ij}=0$ otherwise.  Note that $\aut{A}$ can equivalently be interpreted as the sum of all successful paths between activated pairs of initial and final states for the substring $\slice{w}{i}{i+n}$.

These automata have the property that the power series they recognize are shift invariant in the sense demonstrated in the following Lemma:

\begin{lemma}[Behavior shift invariance]
\label{lem:bidiverging-automata-shift-invariant}
For any bidiverging automaton $\aut{A}$ and biinfinite word $w$ let $\charat{w^{\to k}}{i} := \charat{w}{i-k}$.  Then for all $i,j\in\ints$ and $k,n\in\nats$ we have that $\aut{A}(w,i,n)=\aut{A}(w^{\to k},i+k,n)$.
\end{lemma}

\begin{proof}
Follows directly from the definition and the fact that the activation condition is shift-invariant, as discussed earlier.
\end{proof}

As with converging and diverging automata, the behavior of bidiverging automata is a homomorphism.

\begin{lemma}[Behavior is homomorphism (for bidiverging automata)]
\label{lem:behavior-is-linear-bidiverging}
Given bidiverging automata $\aut{A}$ and $\aut{B}$  over some semiring $S$ and scalar values $\alpha$, $\beta$, $\gamma$, $\delta\in S$ we have that $\behav{\alpha\aut{A}\gamma + \beta\aut{B}\delta} = \alpha\behav{\aut{A}}\gamma + \beta\behav{\aut{B}}\delta$
\end{lemma}

\begin{proof}
Follows straightforwardly from the definitions just as it did for converging and diverging automata, so the proof has been left as an exercise for the reader.
\end{proof}

\section{Bidiverging Power Series}
\label{sec:bidiverging-power-series}

\subsection{Bidiverging Power Series Defined}

In the previous section we introduced bidiverging automata, which are a two-way generalization of diverging automata.  In this section we shall likewise introduce bidiverging power series, which are a two-way generalization of diverging power series.

We again let $S$ be a semiring and $A$ be an alphabet.  We then define $\PSAbiinfinite$ to be the set of all power series over $A^\zeta$ with coefficients in $S^{\ints\times\nats}$, which we shall call \emph{bidiverging power series}.  As before, we shall use function notation, so if $v\in\PSAbiinfinite$, $w\in A^\zeta$, $i\in\ints$, and $n\in\nats$, then $v(w,i,n)$ is equal to position $(i, n)$ of the coefficient on the word $w$.

We endow bidiverging power series with an $S$-semibimodule structure that is consistent with the $S$-semibimodule structure with which we endowed $S^\nats$:  for all $x,y\in\PSAbiinfinite$ we have that addition is given by $x+y := (w,i,n) \mapsto x(w,i,n) + y(w,i,n)$, for all $s\in S$ we have that left-multiplication is given by $sx := (w,i,n) \mapsto s x(w,i,n)$ and that right multiplication is given by $xs := (w,i,n) \mapsto  x(w,i,n)s$, and the additive identity is given by $(w,i,n)\mapsto 0$.  It is easy to see that the semibimodule laws hold, making $\PSAbiinfinite$ an $S$-semibimodule.

\subsection{Rational Bidiverging Power Series}

As with diverging power series, there are two basic ways in which we shall construct bidiverging power series from other kinds of power series.  The first way to build a bidiverging power series from another kind of power series is \emph{infinite iteration}, denoted by $\,^\zeta$, which is defined as follows:  Let $s$ be a proper \emph{converging} power series;  then
$$s^\zeta(w,i,n) := s^*(\slice{w}{i}{i+n}) \cdot \allbiprefixes{w}{s^*}$$
where $\allbiprefixes{w}{x}=1$ if for every $i_0\le j_0$ there exist $i\le i_0\le j_0\le j$ such that $x(\slice{w}{i}{j})\ne 0$, and $\allbiprefixes{w}{x}=0$ otherwise.

The second way to build a bidiverging power series is  \emph{conjoining}, denoted by $\cdot \star \cdot \star \cdot$, which takes three converging power series and forms a bidiverging power series as follows:  Let $x,m,y\in \PSAfinitenoempty$.  Then the  \emph{conjoin} of $x$, $m$ and $y$ is given by,
\begin{align*}
&(x\star m\star y)(w,i,n) := (x^*my^*)\paren{\slice{w}{i}{i+n}} \cdot \allbiprefixes{w}{x^*my^*}.
\end{align*}

Having defined these two ways of building bidiverging power series from other power series, we shall now define \emph{rational} bidiverging power series, $\Rat^\zeta(S,A)$, as the smallest set such that
\begin{enumerate}
\item $\Rat^\zeta(S,A)$ is closed under finite sums;
\item $\Rat^\zeta(S,A)$ is closed under left- and right-multiplication by elements from $S$;
\item for all $z\in\ratnoempty$, $z^\zeta\in\Rat^\zeta(S,A)$; and
\item for all $x,y,m\in\ratnoempty$, $x\star m\star y\in\Rat^\zeta(S,A)$.
\end{enumerate}

As with diverging power series, there is a simple characteristic form for this set, as shown in the following Lemma.

\begin{lemma}[Characteristic representation for bidiverging power series]
\label{lem:characteristic-bidiverging}
A bidiverging power series $p\in \PSAbiinfinite$ is rational if and only if there exist finite index sets $I$ and $J$ and sequences $\{a_i,b_i\}_{i\in I}\subset S$, $\{x_i,y_i,m_i\}_{i\in I}\subset \ratnoempty$, $\{l_j,r_j\}_{j\in J}\subset S$, and $\{z_j\}_{j\in J}\subset \ratnoempty$ such that
$$p = \sum_{i\in I} a_i(x_i \star m_i\star y_i)b_i + \sum_{j\in j} l_j z_j^\zeta r_j.$$
\end{lemma}

\begin{proof}
The proof of this is identical in form to Lemma \ref{lem:characteristic}, so it has been left as an exercise for the reader.
\end{proof}

The set of \emph{recognizable} bidiverging power series, $\Rec^\zeta(S,A)$, is equal to the set of power series that are the behavior of some bidiverging automaton, and as with diverging automata and power series, bidiverging automata and power series are related by a Kleene Theorem.

\begin{theorem}[Kleene's Theorem for bidiverging power series]
\label{thm:kleene-bidiverging}
$$\Rat^\zeta(S,A) =\Rec^\zeta(S,A)$$
\end{theorem}

\subsection{Loopback Automata}

Before proving Kleene's Theorem we will first prove some Lemmas.  As was the case with diverging automata, it will prove useful to start by analyzing some special cases.  The first special form of automaton we shall analyze in the context of bidiverging automata is the loopback automaton.  The most important result we shall prove is an analogue to Proposition \ref{prop:loopback-equals-rational-omega}.

\begin{proposition}[Bidiverging loopback automata recognize the $\,^\zeta$ of rational power series]
\label{prop:loopback-equals-rational-zeta}
The set of power series recognized by bidiverging loopback automata is equal to $\{z^\zeta : z\in\ratnoempty\}$.
\end{proposition}

First, we need a preliminary Lemma, analogous to Lemma \ref{lem:loopback-star-to-omega-conversion}.

\begin{lemma}[Behavior of bidiverging loopback automata]
\label{lem:loopback-star-implies-zeta}
Let $\aut{A}$ be a bidiverging loopback automaton such that its converging counterpart, $\tilde{\aut{A}}$, recognizes the power series $s^*$.  Then $\behav{\aut{A}}=s^\delta$.
\end{lemma}

\begin{proof}
This proof is nearly identical to the proof of Lemma \ref{lem:loopback-star-to-omega-conversion}, save for the difference in the activation condition; given this, the proof has been left as an exercise for the reader.
\end{proof}

We are now ready to prove that loopback automata recognize the $\,^\zeta$ of rational converging power series.

\begin{proof}[Proof of Proposition \ref{prop:loopback-equals-rational-zeta}]
This proof is nearly identical to the proof of Proposition \ref{prop:loopback-equals-rational-omega}, save for the difference in the activation condition and the use of Lemma \ref{lem:loopback-star-implies-zeta} instead of Lemma \ref{lem:loopback-star-to-omega-conversion};  given this, the proof has been left as an exercise for the reader.
\end{proof}

\subsection{Bridge Automata}

The next special form of automaton we shall analyze is what we shall call a \emph{bridge} automaton, which is defined to be an automaton such that there is exactly one initial state and exactly one final state which are not the same state and both have weight 1.  One of the reasons why these automata are special is because they can be formed by \emph{conjoining}, which we now define.  Let $\aut{X}$, $\aut{M}$, and $\aut{Y}$ be normalized automata;  then the conjoin of $\aut{X}$, $\aut{M}$, and $\aut{Y}$ is an automaton denoted by $\aut{X}\star\aut{M}\star\aut{Y}$ which is defined as follows.  First, let $\aut{A}$ and $\aut{B}$ denote respectively the roll of $\aut{X}$ and $\aut{Y}$.  Furthermore let $1$ be the loopback state of $\aut{A}$, $2/3$ be the initial/final state of $\aut{M}$, and $4$ be the loopback state of $\aut{B}$.  Then $\ofaut{Q}{X\star M\star Y} := (\ofaut{Q}{A}\cup\ofaut{Q}{M}\cup\ofaut{Q}{B})/\{2,3\}$, $\ofaut{I}{X\star M\star Y}_i:=\delta_{i1}$,  $\ofaut{F}{X\star M\star Y}_i:=\delta_{i4}$, and
$$
\ofaut{M}{X\star M\star Y}_{ij} :=
\begin{cases}
\ofaut{M}{A}_{ij} & i,j \in \ofaut{Q}{A} \\
\ofaut{M}{M}_{ij} & i,j \in \ofaut{Q}{M}/\{2,3\} \\
\ofaut{M}{M}_{j3} & i\in \ofaut{Q}{M}/\{2,3\}, j = 4 \\
\ofaut{M}{M}_{2j} & j\in \ofaut{Q}{M}/\{2,3\}, i = 1 \\
\ofaut{M}{B}_{ij} & i,j \in \ofaut{Q}{B} \\
0 & \otherwise.
\end{cases}
$$
That is, $\aut{X}$ and $\aut{Y}$ are rolled and merged with $\aut{M}$, with the initial and final states of $\aut{M}$ being merged with the loopback states of respectively $\aut{X}$ and $\aut{Y}$, and the initial and final states being set to the loopback states of respectively $\aut{X}$ and $\aut{Y}$.

The next Lemma shows that conjoining also has the nice property that the behavior of the conjoin is the conjoin of the behaviors.

\begin{lemma}[Behavior of conjoin is conjoin of behavior]
\label{lem:behavior-of-conjoin-is-conjoin-of-star}
Let $\aut{X}$, $\aut{M}$, and $\aut{Y}$ be normalized automata.  Then $\behav{\aut{X}\star\aut{M}\star\aut{Y}}
=\behav{\aut{X}}\star\behav{\aut{M}}\star\behav{\aut{Y}}$.
\end{lemma}

\begin{proof}By Kleene's Theorem \ref{thm:kleene-converging} and Lemma \ref{lem:normalized-rejects-empty-string} we know that there exist $x,m,y\in\ratnoempty$ such that $\behav{\aut{X}}=x$, $\behav{\aut{M}}=m$, and $\behav{\aut{Y}}=y$.  Let $\aut{A}$ and $\aut{B}$ be the respective rolls of $\aut{X}$ and $\aut{Y}$.  By Lemma \ref{lem:loopback-star} we know that $\behav{\aut{A}}=x^*$ and $\behav{\aut{B}}=y^*$.

Now let $w$ be a biinfinite word, $i$ an integer, and $n$ a natural number, and let us consider the value of $(\aut{X}\star\aut{M}\star\aut{Y})(w,i,n)$.  Observe that by construction every successful path has to start on the loopback state in $\aut{A}$ and from there pass through states only in $\aut{A}$ until it lands on the loopback state of $\aut{A}$ for the last time, after which it moves into $\aut{M}$ and passes through states only in there until it eventually it lands on the loopback state of $\aut{B}$, after which it passes only through states in $\aut{B}$ until it ends on the loopback state of $\aut{B}$.  Now consider the set of all paths that land on the loopback state of $\aut{A}$ for the last time on the $j^{\text{th}}$ step and on the loopback state of $\aut{B}$ for the first time after step $j$ on the $k^{\text{th}}$ step.  Because all paths in this set land on the same steps at the loopback state in $\aut{A}$ for the last time and at the loopback state in $\aut{B}$ for the first time after the last time landing on the loopback state in $\aut{A}$, we can factor the sum over all these paths into the product of sums over paths in $\aut{A}$, $\aut{M}$, and $\aut{B}$ for the respective words $\slice{w}{i}{i+j}$, $\slice{w}{i+j}{i+j+k}$ and $\slice{w}{i+j+k}{i+n}$;  since these sums are equal to the value of the behavior at the word for these three automata we therefore have that these three factors are equal to respectively $x^*(\slice{w}{i}{i+j})$, $m(\slice{w}{i+j}{i+j+k})$, $y^*(\slice{w}{i+j+k}{i+n})$.

Let $\aut{C}$ be the converging counterpart of $\aut{X}\star\aut{M}\star{\aut{Y}}$.  Given the discussion above and summing over $j$ and $k$ we have that
\begin{align*}
\aut{C}(\slice{w}{i}{i+n})
&= \sum_{j=0}^n\sum_{k=0}^{n-j}x^*(\slice{w}{i}{i+j})\cdot m(\slice{w}{i+j}{i+j+k})\cdot y^*(\slice{w}{i+j+k}{i+n}) \\
&= (x^*my^*)(\slice{w}{i}{i+n}) \\
\aut{C}(z) &= (x^*my^*)(z).
\end{align*}
(Note that the above sum starts with $k=0$ despite the fact that the initial and the final state are not the same and so the length of the path between them must be greater than zero;  this is okay because for $k=0$ we have that $\slice{w}{i+j}{i+j+k}=\slice{w}{i+j}{i+j}$ is the empty word, and because $m$ is proper it therefore has a zero value coefficient for the empty word.)
Thus, $\aut{C}$ recognizes the power series $x^*my^*$.  In particular this means that every finite word for which the sum of all successful paths in $\aut{C}$ is non-zero is in the support of $x^*my^*$ and vice versa.  The immediate consequence of this is that $w$ activates the initial and final state of $\aut{X}\star\aut{M}\star\aut{Y}$ if and only if for all $i_0\le j_0$ there exists $i\le i_0 \le j_0 \le j$ such that $\slice{w}{i}{j}$ is in the support of $x^*my^*$ and therefore $\allbiprefixes{w}{x^*my^*}=1$.

Putting all of these results together, we see that
$$(\aut{X}\star\aut{M}\star\aut{Y})(w,i,n)=(x^*my^*)(\slice{w}{i}{i+n})\cdot\allbiprefixes{w}{x^*my^*}=(x\star m\star y)(w,i,n).$$

\end{proof}

We now need to define a quasi-inverse operation to conjoining, which we shall call \emph{disjoining}.    Given a bridge automaton $\aut{A}$, the \emph{disjoin} of $\aut{A}$ is a triplet of normalized automata $(\aut{X},\aut{M},\aut{Y})$.  Let $1$ be the initial state of $\aut{A}$ and $2$ be the final state.  Then $\aut{X}$ is the unroll of $\aut{B}$, which is given by $\ofaut{Q}{B}:=\ofaut{Q}{A}$, $\ofaut{I}{B}_i:=\ofaut{F}{B}_i:=\delta_{i1}$, and $\ofaut{M}{B}:=\ofaut{M}{A}$ --- that is, $\aut{B}$ is the result of setting the final state to be the same as the initial state;  $\aut{Y}$ is the unroll of $\aut{C}$, which is given by $\ofaut{Q}{C}:=\ofaut{Q}{A}$, $\ofaut{I}{C}_i:=\ofaut{F}{C}:=\delta_{i2}$, and $\ofaut{M}{C}:=\ofaut{Q}{C}$ --- that is, $\aut{C}$ is the result of setting the initial state to be the same as the final state; and $\aut{M}$ is given by
\begin{align}
\ofaut{Q}{M} &:= \ofaut{Q}{A} \\
\ofaut{I}{M}_i &:= \ofaut{I}{A}_i = \delta_{i1}\\
\ofaut{F}{M}_i &:= \ofaut{F}{A}_i = \delta_{i2} \\
\ofaut{M}{M}_{ij} &:=
\begin{cases}
0 & i = 2 \,\,\text{or}\,\, j = 1 \\
\ofaut{M}{A}_{ij} & \otherwise \\
\end{cases}
\end{align}
that is, the result of deleting the incoming edges on the initial state of $\aut{A}$ and the outgoing edges on the final state.

The sense in which conjoining is a quasi-inverse operation is given in the following Lemma.

\begin{lemma}[Conjoin is inverse of disjoin]
\label{lem:conjoin-quasi-inverse}
Given a bridge automaton $\aut{A}$, and letting $(\aut{X},\aut{M},\aut{Y})$ be equal to the disjoin of $\aut{A}$, we have that $\behav{\aut{X}\star\aut{M}\star\aut{Y}}=\behav{\aut{A}}.$
\end{lemma}

\begin{proof}
The logic here is essentially identical to that used in the proof of Lemma \ref{lem:attaching-and-detaching-are-inverses}.  The only difference is that in this setting we have three copies of $\aut{A}$ and two points at which a path jumps from one copy to another instead of one; in this case successful paths are characterized by the last time the path visits the loopback state in the first copy and the first time the path visits the loopback state in the third copy.  Thus, extending the argument of Lemma \ref{lem:attaching-and-detaching-are-inverses} to work here is left as an exercise for the reader.
\end{proof}

Now we see the significance of bridge automata.

\begin{lemma}[Behavior of bridge automata is the conjoin of rational power series]
\label{lem:bridge-implies-rational}
All power series recognized by bidiverging bridge automata take the form $x\star m\star y$ for some $x,m,y\in\ratnoempty$.
\end{lemma}

\begin{proof}
Suppose we are given a bridge automaton $\aut{A}$.  Let $(\aut{X},\aut{M},\aut{Y})$ be the disjoin of $\aut{A}$, which recall implies that $\aut{X},$ $\aut{M}$ and $\aut{Y}$ are all normalized.  By Kleene's Theorem (Theorem \ref{thm:kleene-converging}) and Lemma \ref{lem:normalized-rejects-empty-string} we know that there exist proper rational converging power series $x,m,y\in\ratnoempty$ such that $\behav{\aut{X}}=x$, $\behav{\aut{M}}=m$, and $\behav{\aut{Y}}=y$.  By Lemma \ref{lem:conjoin-quasi-inverse} we know that $\behav{\aut{A}}=\behav{\aut{X}\star\aut{M}\star\aut{Y}}$, and by Lemma \ref{lem:behavior-of-conjoin-is-conjoin-of-star} we know that $\behav{\aut{X}\star\aut{M}\star\aut{Y}}= \behav{\aut{X}}\star\behav{\aut{M}}\star\behav{\aut{Y}}=x\star m\star y$.
\end{proof}

\subsection{Proof of the Kleene Theorem}

We are almost ready to prove our Kleene theorem, but there is one Lemma left.

\begin{lemma}[Decomposition into bridge automata and loopback automata]
\label{lem:decompose-into-bridge-and-loopback}
For all bidiverging automata $\aut{A}$ there exists a decomposition into a weighted sum of bridge automata and loopback automata, i.e. a tuple $(K,$ $\{l_k,r_k\}_{k\in K},$ $\{A_k\}_{k\in K})$ such that $\behav{\aut{A}}=\sum_{k\in K} l_k \behav{\aut{A}_k} r_k$ where $K$ is an index set, $\{l_k,r_k\}_{k\in K}$ is an indexed set of coefficients in the underlying semiring $S$, and $\{\aut{A}_k\}_{k\in K}$ is an indexed set of automata each of which is a bridge automaton or a loopback automaton.
\end{lemma}

\begin{proof}
This proof has the exact same form as Lemma \ref{lem:decompose-into-loopback-with-or-without-prelude}, but with the role of loopback automata with prelude replaced by bridge automata (which actually simplifies the proof since the addition of a new state to act as the initial state in the $p\ne q$ case is no longer needed as bridge automata have no restrictions on the edges of the initial and final states), and the use of Lemma \ref{lem:loopback-star-to-omega-conversion} replaced by use of Lemma \ref{lem:decompose-into-bridge-and-loopback}.  Given this, the details have been left as an exercise for the reader.
\end{proof}

Finally we are ready to prove our Kleene Theorem for bidiverging power series.

\begin{proof}[Proof of Theorem \ref{thm:kleene-bidiverging}]
First, assume we have been given an automaton $\aut{A}$.  By Lemma \ref{lem:decompose-into-bridge-and-loopback} we know that there exists a decomposition of $\aut{A}$ into a weighted sum of bridge automata and loopback automata.  By Lemma \ref{lem:bridge-implies-rational} and Proposition \ref{prop:loopback-equals-rational-zeta} we know that both kinds of automata have rational behaviors, so because a weighted sum of rational bidiverging power series is also rational we have that $\behav{A}\in\Rat^\delta(S,A)$.

Now assume that we have been given a rational bidiverging power series $p$.  By Lemma \ref{lem:characteristic-bidiverging} we know that $p=\sum_{i\in I} a_i(x_i \star m_i\star y_i)b_i + \sum_{j\in J} c_j z_j^\zeta d_j$ for some sequence of rational diverging power series $\{x_i, m_i, y_i\}_{i\in I}\subset\ratnoempty$, some sequences of semiring elements $\{a_i,b_i\}_{i\in I}\subset S$ and $\{c_j,d_j\}_{j\in J}\subset S$, and some sequence of rational converging power series $\{z_j\}_{j\in J}\subset\ratnoempty$.  By Corollary \ref{cor:kleen-normalized} we know that for every $i\in I$ there exist normalized converging automata $\aut{X}_i$, $\aut{M}_i$ and $\aut{Y}_i$ such that $\behav{\aut{X}_i}=x_i$, $\behav{\aut{M}_i}=m_i$ and $\behav{\aut{Y}_i}=y_i$, and by Lemma \ref{lem:behavior-of-conjoin-is-conjoin-of-star} we know that $\behav{\aut{X}_i\star\aut{M}_i\star\aut{Y}_i}= \behav{\aut{X}_i}\star\behav{\aut{M}_i}\star\behav{\aut{Y}_i}=x_i\star m_i \star y_i$.  By Proposition \ref{prop:loopback-equals-rational-zeta} we know that for every $j$ there exists a bidiverging automaton, $\aut{Z}_j$, that recognizes the power series $z_j^\zeta$.  Let $\aut{A}:=\sum_{i\in I}a_i (\aut{X}_i\star\aut{M}_i\star\aut{Y}_i)b_i + \sum_{j\in J}c_j \aut{Z}_jd_j$, and because by Lemma \ref{lem:behavior-is-linear-bidiverging} the behavior operation is a homomorphism, we see that $\behav{A}=p$.

\end{proof}

\section{Application: Quantum Simulation}
\label{sec:quantum-simulation}

\subsection{Background for Finite Systems}

In the previous sections we have presented formalisms for diverging and bidiverging automata, but we have not shown how they can be applied to model relevant systems in quantum physics.  We shall do so in this section.  First, though, we need to introduce some basic concepts from (discrete\footnote{It is possible to apply similar ideas to systems with continuous degrees of freedom --- see Ref. \cite{PhysRevLett.104.190405} for an example --- but that is outside the scope of this discussion.}) quantum physics.  

We shall define a quantum system to be a finite (for now) set of configurations $A$.  At any time it will be in a state, usually denoted by $\psi$, which is a \emph{superposition} of these configurations, by which we mean that $\psi\in\powerseries{\cmps}{A}$.

Before proceeding, it is useful to define what it means to take the \emph{dual} of $\psi$.  The dual operation for quantum states is denoted $\,^\dagger\in \powerseries{\cmps}{A}\to\paren{\powerseries{\cmps}{A}\to \cmps}$ and is defined by $(\sum c_i a_i)^\dagger =\sum c_i^* a_i^{-1}$ where $c_i\in\cmps$, $a_i\in A$, and for all $a_j,a_k\in A$ we have $a_j^{-1}(a_k) = \delta_{jk}$.  Given the dual operation, we define the \emph{normalization} of $\psi$ as $\norm{\psi}^2 := \psi^\dagger(\psi)$.  The quantity $|\psi(a)|^2/\norm{\psi}^2$ gives the probability of observing the configuration $a$ if the system is measured (in the $A$ basis).\footnote{Sometimes when dealing with finite systems it is simply assumed that the state is normalized and so there is a burden to ensure that all manipulations of the state preserve this property, but we take the other common approach of simply not worrying about the normalization as it can always be accounted for at the end of the computation.}  After measurement, the state of the system is said to have been collapsed into configuration $a$ as at that point $\psi=a$.

Part of what makes quantum mechanics interesting is that there is more than one way to measure a quantum system.  For example, consider a single particle with a quantum spin which can be in the `up' configuration along the Z-axis, denoted by $\uparrow$, or in the `down' configuration along the Z-axis, denoted by $\downarrow$, so that $A := \{\uparrow,\downarrow\}$.  Possible states of this system include $\uparrow$, $\!\uparrow + \downarrow$, $\frac{1}{\sqrt{3}}\!\!\uparrow -\,i\!\downarrow$ and so on.  Measuring the spin of the system along the Z-axis will collapse the state of the system into either $\uparrow$ or $\downarrow$, but interestingly measuring the system along the X-axis will collapse the state of the system into either $\uparrow+\downarrow$ or $\uparrow-\downarrow$, and measuring the system along the Y-axis will collapse the state of the system into either $\uparrow+\,i\!\downarrow$ or $\uparrow-\,i\!\downarrow$.\footnote{For the interested reader we mention in passing that this is an example of the uncertainty principle in action:  By measuring along the X-axis we collapse the state of the system into the form $\uparrow \pm \downarrow$, which causes $\uparrow$ and $\downarrow$ to have equal amplitude and hence to have equal probabilities if we measure along the Z axis.  So although we now know the spin along the X-axis, we have maximally prevented ourselves from knowing what we will get if we measure the spin along the Z-axis.  We get an analogous effect if we measure along other axes, and hence we conclude that exact knowledge of one axis ensures maximal uncertainties of the other axes.}

Now, when we measure an observable quantity of the system we do not usually get an exact reading of the state of the system but rather there is some dial that we read that gives us a real number from which we can infer partial or total information about the state of the system.  For example, when we measure the spin of a particle we might do so by sending it through a special magnetic field that deflects it upward or downward based on its spin, and then measure by how much it is deflected, with $+1$ corresponding to `up' and $-1$ corresponding to `down'.  For this reason, an observable quantity consists of two pieces of information:  the indexed set of possible states to which the system might be collapsed, $\{\psi_i\}_{i\in I}$, and, for each $i\in I$, the value $\lambda_i\in\reals$ that will be observed if the system collapses into that state.\footnote{If multiple configurations have the same value then if that value is measured the system has collapsed into some (unknown) superposition of these configurations.}  Both of these pieces of information can be stored within a single operator $O := \sum_{i\in I} \lambda \psi_i \psi_i^\dagger$ that is an endomorphism (linear operator) over the configuration space $\powerseries{\cmps}{A}$ with the property that for every $i\in I$ we have that $O(\psi_i) = \lambda_i\cdot \psi_i$ --- that is, $O$ has an \emph{eigendecomposition} into \emph{eigenvalues}, $\{\lambda_i\}_{i\in I}$, and associated \emph{eigenvectors} or \emph{eigenstates}, $\{\psi_i\}_{i\in I}$.  Given that for all $x$ and $y$ we have that $x^\dagger(y)=y^\dagger(x)^*$ (which follows directly from the definition of $\dagger$ and the fact that complex numbers commute), it is not hard to see that $O$ is \emph{self-adjoint}, which means that for all $x$ and $y$ we have that $(x^\dagger\circ O)(y)^* = (y^\dagger\circ O)(x)$.  We shall say that $O$ lives in the space $\powerseries{\cmps}{A\to A}$ where, for all $a_i\to a_j\in (A\to A)$ and $a_k\in A$ we have that $(a_i\to a_j)(a_k) = a_j \delta_{ik}$.

With the observable $O$ (by which we shall mean the \emph{operator} representation of the observable described above) in hand, and given an arbitrary state $\psi$, there is also another useful piece of information we can calculate which is the \emph{expected value} of $O$, given by $(\psi^\dagger\circ O)(\psi)$;  this gives us the average value that we would expect to see over repeated experiments with the system reinitialized to $\psi$ each time.    Part of the reason that this quantity is so important is because we do not always know the eigenvalue decomposition of $O$ and so, for example, randomly generating many states and computing the expected value of $O$ can provide estimates of the maximum and minimum values of the observable.

Three observables appear so often that they are worth mentioning here;  they are the three \emph{Pauli spin matrices}, $X$, $Y$, and $Z$, which correspond to the observables for the spin along the respective X-, Y-, and Z-axes.  Recalling that the eigenstates of $X$ were (after normalizing) $\frac{1}{\sqrt{2}}(\uparrow+\downarrow)$ and $\frac{1}{\sqrt{2}}(\uparrow-\downarrow)$, we see that $X= \frac{1}{2}([\uparrow+\downarrow]\to[\uparrow+\downarrow]) - \frac{1}{2}([\uparrow-\downarrow]\to[\uparrow-\downarrow])=(\uparrow\to\downarrow) + (\downarrow\to\uparrow)$, where we obtained the shorter form by taking advantage of the fact that the $\to$ operator is bilinear and so we can expand the longer form and eliminate the terms that cancel.  Following a similar process for the other operators we obtain $Y=-i(\uparrow\to\downarrow) + i(\downarrow\to\uparrow)$ and $Z=(\uparrow\to\uparrow)-(\downarrow\to\downarrow)$.

Another observable that is very important is the \emph{hamiltonian}, as it both defines the energy observable of the systems and also completely specifies how a state evolves over time in the following way:  If $H=\sum_i E_i \psi_i \psi_i^\dagger$ is the hamiltonian of a system (where the values $E_i$ are the energies) then $U(\Delta t)=\sum_i e^{-i E_i \Delta t} \psi_i\psi_i^\dagger$ \footnote{In principle the first factor should be $e^{-i (E_i/\hbar)\Delta t}$ where $1/\hbar$ is effectively a unit conversion factor from energy to temporal frequency, but it proves convenient in many contexts (such as this one) to simply assume that we are working in a system of units such that $\hbar=1$.} is the endomorphism that takes an arbitrary starting state and maps it to the state of the system after $\Delta t$ time has passed.  (Note that this operator is independent of the starting time.)

\subsection{Application to Biinfinite Systems}

Up to now we have assumed that we are working with a finite system, but it is often incredibly useful to study systems that are \emph{infinite} in extent.  The reason for this is that it gets rid of the boundaries on the sides of the system by making them be infinitely far away.  This allows us to study the \emph{bulk} behavior of the system without having the boundary effects mixed in.  This is useful not only because it makes it easier to understand what is going on by isolating out one of the kinds of behavior, but also because real-life systems tend to be almost `infinitely large' given that they have on the order of Avogadro's number of particles ($\approx 6.02\times 10^{23}$) so that the vast majority of the material behaves as if it were in an infinitely large system.

Thus, we now say that the set $A$ contains the set of configurations not for the full system, but only for a single \emph{site} of the system.  At this point we are going to assume that we are studying a system in a single dimension.  Obviously this is being done right now because it connects with the formalism presented in this paper, but it is also the case that the study of one-dimensional systems in physics is quite common.  There are a couple of reasons for this.  First, systems with multiple dimensions are still very difficult, and so one-dimensional versions of a system give an approach that may glean some useful information, or at the very least provide a useful eventual contrast that shows how phenomena change when the number of dimensions increase.  Second, there are many real-life systems that can be treated as being one-dimensional for various reasons, such as narrow tubes where the interactions not along the axis are negligible.

So given that we have a one-dimensional biinfinite system, its configurations are given by $A^\zeta$, and naively its state space would be $\powerseries{\cmps}{A^\zeta}$, but the problem with this space is that computing the normalization and the expected value of an observable are in general not possible as the sums won't converge.  Thus, we must find a subspace within this space such that we can \emph{make} them converge.  One possibility is to work within a von Neumann tensor product space (also called an ``incomplete'' tensor product space) which is essentially the maximal subspace of $\powerseries{\cmps}{A^\zeta}$ that is a Hilbert space (see Ref. \cite{Neumann1939}).  Unfortunately this subspace is restrictive and does not allow us to use many basic but important operators such as $I\star Z \star I$, which is used to define a magnetic field or to measure the magnetization.

Fortunately this entire paper has described an alternative solution to this problem --- rational bidiverging power series.  That is, we let the state space live in $\Rat^\zeta(\cmps,A)\subset \powerseries{\cmps^{\ints\times\nats}}{A^\zeta}$.  Because of this, all states have equivalent representations as automata which give us efficient ways to compute representations of the normalization and expected values.

To define how to calculate these values, we first shall first define how transducers work.  First, let $\Aut(A)$ be the set of bidiverging automata over $A$ and $\cmps$.  Now let $\aut{O}\in\Aut(A\to B)$ and $\aut{A}\in\Aut(A)$ be bidiverging automata.  Then $\aut{O(A)}$ is given by
\begin{align*}
\ofaut{Q}{O(A)} &:= \ofaut{Q}{O}\times\ofaut{Q}{A} \\
\ofaut{I}{O(A)}_{(i,j)} &:= \ofaut{I}{O}_i \ofaut{I}{A}_j \\
\ofaut{F}{O(A)}_{(i,j)} &:= \ofaut{F}{O}_i \ofaut{F}{A}_j \\
\Mofautwithsup{O(A)}{b}_{(i,j),(k,l)} &:=\sum_{a\in A} \Mofautwithsup{O}{a\to b}_{ik} \Mofautwithsup{A}{a}_{lm}.
\end{align*}

Now that we have transducers, we define the dual operation as simply mapping every $a\in A$ to $0\in\{0\}$, i.e. so that if $\aut{A}\in\Aut(A)$ then $\aut{A}^\dagger\in\Aut(A\to\{0\})$.  In particular, if $\aut{B}\in\Aut(A)$ then $\aut{C}:=\aut{A}^\dagger(\aut{B})\in\Aut(\{0\})$.  Thus we see that the dual transducer of an automaton has the effect of essentially mapping all words to scalars, just as the finite definition did.  Note that the input language of  $\aut{C}$ is $\{0^\delta\}$, which both consists of only a single string and is completely invariant under shifts, so we can effectively ignore the word and position arguments and treat $\aut{C}$ as a map from natural numbers to complex numbers, i.e. $\aut{C}(w,i,n)\equiv \aut{C}(n)$ for all $w\in\{0^\delta\}$ and $i\in\ints$.

Because bidiverging automata directly correspond to bidiverging power series, all of the operations we have just defined can be lifted to act on bidiverging power series.  Specifically, for any power series $a,b\in\Rat^\zeta(\cmps,A)$ we define $a^\dagger(b):\nats\to\cmps$ by $a^\dagger(b) := n \mapsto \aut{A}^\dagger(\aut{B})(n)$ where $\behav{\aut{A}}=a$ and $\behav{\aut{B}}=b$, and if $o\in\Rat^\zeta(\cmps,A\to B)$ then we define $o(a) := \behav{\aut{O}(\aut{A})}$ where $\behav{\aut{O}}=o$.  These operations automatically give us a well-defined normalization, but for expected values it is useful to clarify that the ratio of the two maps should be taken pointwise --- that is, given an endomorphism $o\in\Rat^\zeta(\cmps,A\to A)$, the expected value of $\psi\in\Rat^\zeta(\cmps,A)$ is defined to be
$$n\mapsto \frac{(\psi^\dagger\circ o)(\psi)(n)}{\psi^\dagger(\psi)(n)}.$$

At this point it might not be obvious how much we have gained.  It is true that we have found a subset of bidiverging power series where normalizations and expected values are well-defined, but in the process we have paid three prices:  first, we have required that our operators live in $\Rat^\zeta(\cmps,A\to A)$, second, our observable values are now sequences rather than real values, and third, our states are forced to live in a restricted space that will in general not contain the actual physical states.   Fortunately, the first price turns out to be fairly low one because most operators that physicists care about turn out to be exactly representable as bidiverging power series over endomorphisms.  For example, the average magnetization of a system is given by the sum over terms where every term has the $Z$ operator at one site and the identity ($I$) at the rest so that the bidiverging power series takes the form $I\star Z \star I$.  (This might look at first like it only generates a single term, but in fact the resulting power series accepts any shift of the string $I^{\tilde\omega} Z I^\omega$, so it acts like a sum over all operators with $Z$ at a single site and $I$ at the rest.)  Furthermore, those interactions that can't be represented exactly can usually be approximated fairly well, as we shall see later.

The second seeming price --- that of having a sequence in the place of a scalar value --- is actually a boon instead of a bane.  Consider, for example, the energy of an infinite system.  Obviously the total energy of the system is going to be infinite in general, but knowing this is not particularly helpful.  What \emph{is} helpful instead is knowing how the energy grows with the size of the system, and this is exactly the information that is encoded in the sequence!  That is, because the value of the expected value at position $n$ is exactly equal to sum over all paths that have length $n$, it naturally has the interpretation as the expected value of any collection of $n$ contiguous sites of the system in the absence of boundary effects, or alternatively as the component of the expected value in a system of size $n$ that is due to bulk behavior rather than boundary effects.

For an example of how the expected value works in this way, consider a quantum system with $A = \{\uparrow,\downarrow\}$ that is in the state $\psi = \,\,\uparrow^\zeta$.  It is left as an exercise for the reader to show that $|\psi|=1$.  Now let $O:=I \star Z\star I$ be the magnetization observable discussed earlier.  The automaton for $O$ is $\aut{O}$ which is given by $\ofaut{Q}{O}:=\{1,2\}$, $\ofaut{I}{O}_i :=\delta_{i1}$, $\ofaut{F}{O}_i := \delta_{i2}$, and
$\ofaut{M}{O}_{ij} = I\cdot\delta_{ij} + Z\cdot\delta_{i1}\delta_{j2}$.
Because the state is normalized, the expected value is given by $(\psi^\dagger\circ O)(\psi)$, and it is left as an exercise for the reader to show that this is equal to the behavior of an automaton $\aut{E}$ given by $\ofaut{Q}{E}=
\ofaut{Q}{O}$, $\ofaut{I}{E}=\ofaut{I}{O}$, $\ofaut{F}{E}=\ofaut{F}{O}$, and $\ofaut{M}{E}_{ij}=1-\delta_{i2}\delta_{j1}$, and that based on this, the expected value of $O$ with respect to $\psi$ is $n \mapsto n$.  This makes perfect sense because every time a spin pointing up has been added to the system we would expect the magnetization to grow by a single unit.  The main physical quantities of interest --- including magnetization and energy --- tend to be \emph{extensive} quantities, which means that they grow linearly with the size of the system and hence have an interpretation as an energy or magnetization per site.  In general, though, the result of an expected value will not be linear, but it is restricted to have the following general form,
$$n\mapsto\frac{\sum_{i\in I} \lambda_i^n \cdot\text{poly}_i(n) + \sum_{j\in J} c_j \delta_{nj}}{\sum_{k\in K} \lambda_{k}^n \cdot\text{poly}_{k}(n) + \sum_{l\in L} c_{l}\delta_{nl}}$$
where $\text{poly}_i(n)$ denotes some polynomial in $n$.  This follows from the fact that every matrix --- and therefore the transition matrix for the expected value automaton in particular --- is similar to a matrix in Jordan Normal Form, and it can be shown that raising a matrix in Jordan Normal Form to an integer power $n$ results in a new matrix where every component has the form $\sum_{i\in I} \lambda_i^n \cdot\text{poly}_i(n) + \sum_{j\in J} c_j \delta_{nj}$ for some $I$, $\{(\lambda_i,\text{poly}_i)\}_{i\in I}$, $J$, and $\{c_j\}_{j\in J}$.  For the details, see pages 385 and 386 of Ref. \cite{Horn1991}, and specifically let $p(\lambda)=\lambda^n$ to obtain the above result.

Finally, the third price --- the fact that we are living in a restricted space --- is not a deal breaker as long as the states in this space provide sufficiently good approximations to the physical states of interest, and fortunately it turns out that they do in practice;  see the end of this section for an example of a simulation that illustrates this.

\subsection{Simulation Methodology}

We have now established that bidiverging power series provide a means of approximating quantum states in a manner that has well-defined and useful values for expected values of observables, but this fact would be uninteresting if there were not ways to find sufficiently good approximations of quantum states of interest.  Fortunately, there are such ways, and they take advantage of the fact that when we study quantum systems we are often most interested in the \emph{ground} state or states --- that is, the lowest energy state or states --- and possibly the \emph{excited} states --- that is, those states just above the ground state energy.  The reason for the focus on these states is that they tend to have the most interesting behavior because as the energy grows higher the system acts increasingly like a classical system with classical properties rather than a quantum system with quantum properties.

So given that we are interested in the lowest energy states, a natural approach is to start by finding the ground states and then to work our way up from there.  To find a ground state, we take advantage of the fact that the expected value of the energy will never be less than the ground state energy, and furthermore the lower the expected value is the closer we are to the ground state energy and thus hopefully (but not necessarily) a ground state.  Thus, a heuristic that turns out to be effective in practice (although it is of course not guaranteed to work\footnote{One reason why this might not work is because the structure of the energy eigenstates is such that there are states that have energy almost equal to the ground state energy but which are not within easy reach of a ground state due to the presence of so-called \emph{forbidden transitions}.  This does not tend to cause problems in practice, but interestingly it does cause problems for systems that are designed such that, say, the ground state encodes the solution of an NP-complete problem, which is why engineering such systems then cooling them down as close as possible to absolute zero does not actually work as a method for solving NP-complete problems.}) is to start by making some ansatz for the ground state --- say, that it takes the form of a bidiverging power series --- and then to adjust the free parameters to minimize the expected value of the energy.  Once the energy has been minimized we take the resulting state to be a ground state (or at least, a sufficiently good approximation of it) and from there one can in principle find the next lowest state (possibly another ground state) by performing the same procedure but with a constraint that the new state must be orthogonal to the old state.  This method is known in the field of physics as the \emph{variational} method.

For bidiverging power series, there are a couple of basic variational approaches one can use.  First, there is the \emph{imaginary time evolution} approach (see Refs. \cite{Vidal2007} and \cite{Orus2008}).  To understand how this works, it is useful to first recall that the operator that evolves a given state forward in time by $\Delta t$ units is given by $U(\Delta t)=\sum_i e^{-i \Delta t E_i \psi_i\psi_i^\dagger}$, where the $E_i$ are the energy eigenvalues and the $\psi_i$ are the associated energy eigenstates.  Now observe what happens if we feed an \emph{imaginary} time into $U$:  $U'(\Delta t) := U(i\Delta t) = \sum_i e^{-\Delta t E_i \psi_i\psi_i^\dagger}$.  The new function $U'$ has the effect of causing each energy eigenstate component of the state to decay at a rate exponentially proportional to its energy, so by evolving a state arbitrarily far forward in imaginary time the proportion of the state that is in the ground states can be made arbitrarily high, giving us a means of obtaining a very good approximation of a ground state from a random initial state.\footnote{Actually, if by some horrible accident we start with a state that has \emph{zero} overlap with any ground state then this is not true, but this is a low probability event and furthermore it can be mitigated by trying several initial random states and keeping the lowest energy one.}  The primary difficulty with this method is that in general we do not actually have a means of applying $U(\Delta t)$ exactly --- in fact, if we did then we most likely know or can easily obtain the eigendecomposition and hence have no need for a variational approach in the first place.  Fortunately, it turns out that there is an approximation known as the Suzuki-Trotter expansion\footnote{Trotter figured out a first order approximation in Ref. \cite{Trotter1959}, and Suzuki generalized this idea to generate approximations at all orders in Ref. \cite{Suzuki1976}.} that allows one to systematically approximate $U$ in terms of polynomials of $H$.  The approximation can be taken to arbitrary order.  For example,  to first order in the size of the time step we have that $U'(\delta t) = I - \delta t\cdot H + O(\delta t^2)$.  By combining many small time steps we have that $U'(\Delta t)=U'(\delta t)^{\Delta t/\delta t}\approx(I-\delta t\cdot H)^{\Delta t/\delta t}$,  the total error to first order of which is proportional to $\Delta t/\delta t \cdot \delta t^2 = \Delta t\cdot \delta t,$ which can be made arbitrarily small for any $\Delta t$.  Thus, in practice the imaginary time evolution approach involves picking an order for the Suzuki-Trotter decomposition (higher order means more calculations and greater complexity per step but fewer steps), picking a $\delta t$, and then repeatedly applying this approximation of $U'(\delta t)$ until the number of states in the automaton grows unmanageably large, at which point a truncation operation is applied that attempts to find the best possible approximation to the original automaton that uses fewer states.  The process of alternating between applying a small time step and truncating the automaton to prune it to a manageable size is then continued until the state converges to a fixed point.

There is another approach that uses \emph{sweeping} (see Refs. \cite{McCulloch2008} and \cite{Crosswhite2008}).  The basic idea behind this approach is that rather than applying a global transformation to the whole system until we converge to an answer we instead zoom in on a specific site and optimize it independently from the rest of the system.  We do this by constructing an \emph{environment} for the focused site that effectively takes the expected value of the hamiltonian for the infinite system and sums over all sites but the focused site.  The end result is a matrix $M$ such that $(\psi_f^\dagger\circ M)(\psi_f)=(\psi^\dagger\circ H)(\psi)$ is equal to the expected value of the hamiltonian for the state of the entire system $\psi\in\Rat^\zeta(\cmps,A)$ as a function of the focused site $\psi_f\in \powerseries{\cmps^{Q^{\psi}\times Q^{\psi}}}{A}$.  Because $M$ is a small, bite-sized matrix, we can (relatively) easily solve for its lowest energy eigenstate\footnote{Technically there should also be a similar matrix $N$ obtained by summing over all other sites for the normalization and we should be solving the generalized eigenvalue problem $Mv = \lambda Nv$, but in practice it turns out that we can keep the system normalized in such a way that $N$ is the identity.} and substitute it for $\psi_f$, thus reducing the energy of the entire system.  After doing this, we then \emph{absorb} $\psi_f$ to the left or the right by making a copy of it and expanding the respective left or right sum in the environment to include it.  We then repeat this process until we have converged to a fixed point, and then we increase the number of states in the automaton and then repeat the whole process until a fixed point has been reached (or we run out of memory).

Once a ground state has been found, the ability to compute expectation values means that one can perform many kinds of analyses on it.  For example, we can compute the magnetization, and we can also compute a \emph{correlator}, which is an operator of the form $I \star ZI^k Z\star I$ that provides information about how likely a particle at some arbitrary site $i$ is to agree with the particle at site $i+k+1$ if both particles have their spin measured along the $Z$ axis.

\begin{figure}
\includegraphics[width=\textwidth]{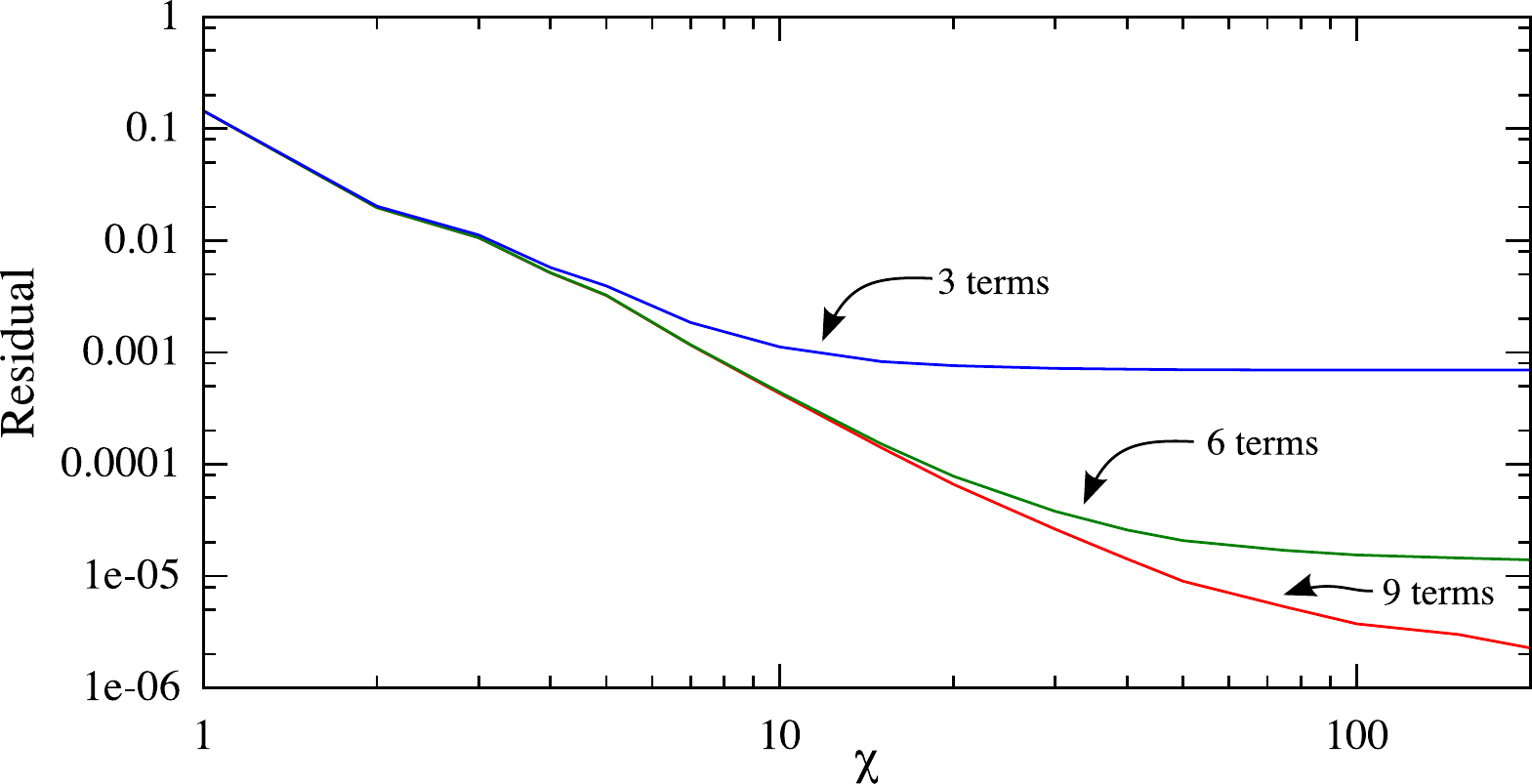}
\caption{\label{fig:example-residuals}This plot shows the energy residual (the difference between the exact and approximated energies) as a function of $\chi$ (the number of states used in the automaton, which increased over time as the solver ran) for the simulations run using the 3-term, 6-term, and 9-term expansions of the hamiltonian.  [Note: This figure was taken directly from Ref. \cite{Crosswhite2008} for the sake of illustration; it was originally created by the author of this paper.]}
\end{figure}

\begin{figure}
\includegraphics[width=\textwidth]{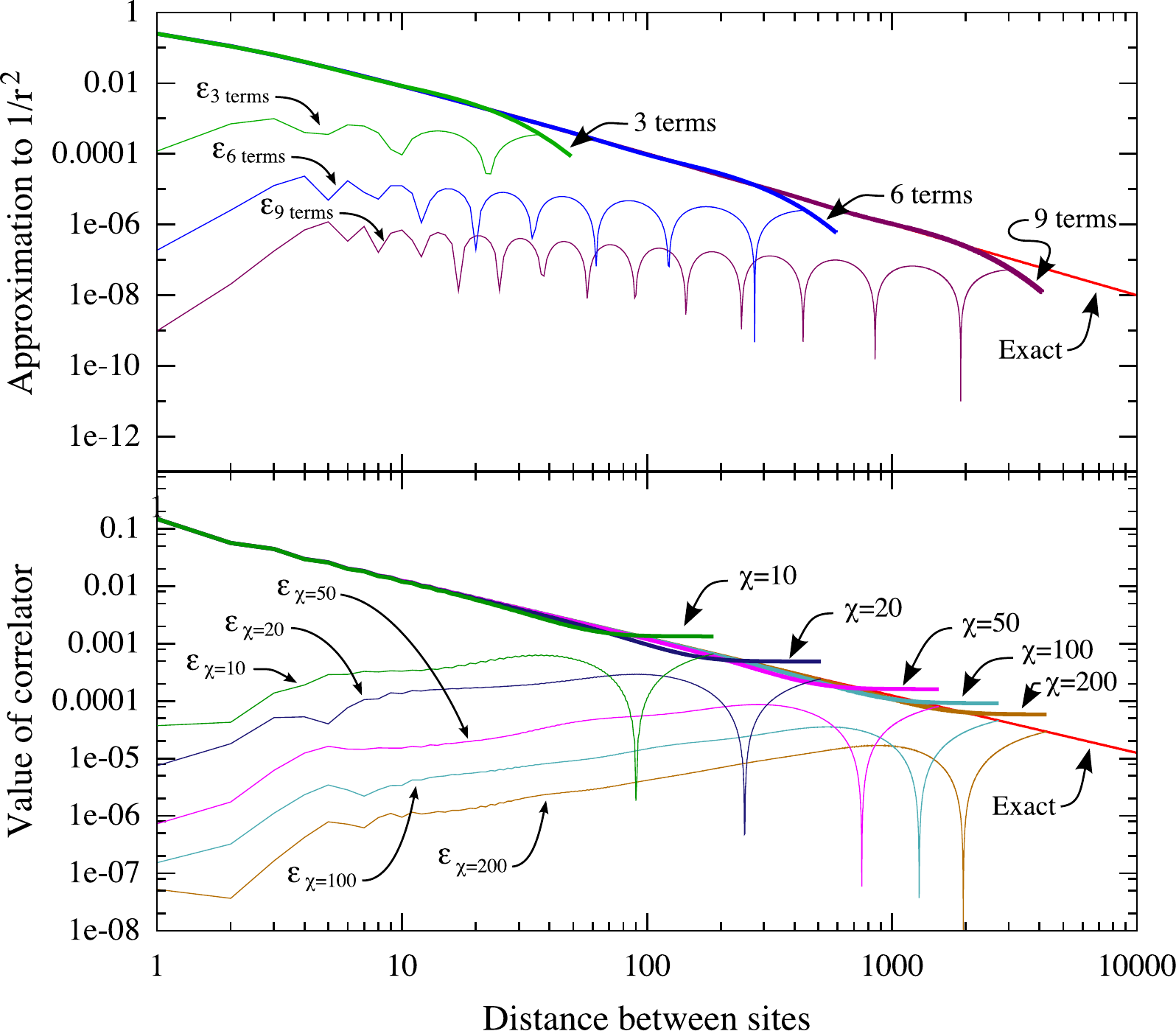}
\caption{\label{fig:example-combined}The top of this figure plots the expansions of $1/r^2$ using 3, 6, and 9 terms, against the exact value of $1/r^2$.  The bottom of this figure plots the correlator (which can be thought of how likely it is that two spins will agree as a function of distance) for states with various values of $\chi$ (the number of states in the automata).  In both cases, the curves below the main curves that are tagged with epsilons are the residuals (the differences between the approximate values and the exact values).  [Note: This figure was taken directly from Ref. \cite{Crosswhite2008} for the sake of illustration; it was originally created by the author of this paper.]}
\end{figure}

\subsection{Proof of Concept}

To illustrate an example of simulating a quantum system, we consider the Haldane-Shastry model (see Refs. \cite{Haldane1988} and \cite{Shastry1988}), which was simulated using the sweep method we just discussed in Ref. \cite{Crosswhite2008}.  This model is interesting for two reasons:  first, it is exactly solvable, so that we can see how well the obtained ground state emulates the properties of the true ground state, and second, it involves a hamiltonian with a sufficiently non-trivial structure that the model provides a non-trivial test for the approaches we have been discussing.  The hamiltonian of the Haldane-Shastry model takes the form, $H = \sum_{i=-\infty}^{+\infty}\sum_{r=1}^\infty \vec{\sigma}_i\cdot\vec{\sigma}_{i+r}/r^2$ where $\vec{\sigma}_i=(X,Y,Z)$ acting on site $i$ (where $X$, $Y$, and $Z$ were defined earlier).  This model physically represents a biinfinite chain of particles with spins that interact antiferromagnetically (that is, so that they don't want to line up) with each other along all directions and with a potential that decreases with the square of the distance.  Now, this hamiltonian turns out to be one of the rare cases we mentioned which cannot be expressed exactly as a bidiverging power series due to the $1/r^2$ coefficient.  Fortunately it can be expressed arbitrarily well by using a sum of decaying exponentials, i.e. $\sum_i \alpha_i \beta_i^r$ for some $\alpha_i$ and $\beta_i$.  In Ref. \cite{Crosswhite2008} we computed approximation using 3 terms, 6 terms, and 9 terms, and Figure \ref{fig:example-combined} (top) shows that the approximation works reasonably well in practice, as for 9 terms it produces an approximation that has an error less than about $10^{-6}$ for distances up to 3000 sites\footnote{To get a sense of why this number is usefully large, it is helpful to know that most models that are studied only consider interactions between nearest neighbors or possibly next-nearest neighbors, and $3000\gg 2$.}.  Given $\{\alpha_i,\beta_i\}_{1\le i\le N}$ for an $N$ term expansion, the final (approximate) hamiltonian took the form
$$\sum_{i=1}^N \alpha_i\paren{[I\star X(\beta_i I)^* X\star I] + [I\star Y(\beta_i I)^* Y\star I] + [I\star Z(\beta_i I)^* Z\star I]}$$

We applied the sweeping approach discussed earlier to each of these approximated hamiltonians;  for each value of $\chi$, which is what we denoted the number of states in the automaton, we computed the energy.  The expected value of the energy turned out to be a linear function\footnote{More precisely, we found that \emph{in the large $n$ limit} the expected value of the energy turns out to be a linear function, which was sufficient for our purpose of comparing it to the exact energy per site of the Haldane-Shastry model in the infinite size limit.  We computed only the large $n$ limit of the energy because computing the full function would have required computing the full Jordan Normal Form of the expected value's automaton's transition matrix, which would have been expensive.} and hence could be interpreted as an energy per site, which matches the exact solution of the model.  The error in the energy per site of the approximate solution (as obtained by comparing it to the exact solution) for each of the three approximations of $H$ and for each value of $\chi$ is plotted in Figure \ref{fig:example-residuals};  in particular we see that the solution obtained using the 9-term approximation of $H$ had an energy residual of only about $3\times 10^{-6}$ for $\chi=200$.  We also computed the correlator for the solution obtained using the the 9-term approximation and plotted it against the exact value of the correlator in Figure \ref{fig:example-combined} (bottom);  in particular we see that for $\chi=200$ the correlators match to within about $5\times 10^{-5}$ out to 3000 sites.  This example has demonstrated that the techniques that have been discussed throughout this section do work in practice, allowing us to obtain and analyze very good approximations to the ground states of biinfinite systems.

\section{Conclusions}
\label{sec:conclusions}

In this paper we have introduced a new kind of automaton called a \emph{diverging automaton} which explicitly captures the divergences caused by uniting infinite words with weighted automata by modeling the divergence as a sequence of weights.  We have presented a corresponding \emph{diverging power series} as well as natural rational operations, and proven a Kleene Theorem that shows that the set of rational diverging power series is equal to the set of behaviors of diverging automata.  We have furthermore presented extensions of these ideas to biinfinite words, resulting in \emph{biinfinite automata} and \emph{biinfinite power series} with, of course, another Kleene Theorem connecting the first to the rational subset of the second.  Finally, we have demonstrated the usefulness of these constructions by showing how rational bidiverging power series are very important in quantum simulation due to their ability to provide a powerful means of approximating the states of biinfinite quantum systems.

There are at least two obvious directions for future research.  First, it would be good to find a theory that generalizes and unites the theory we have just presented here with the theory of Conway $\,^*$-semiring---$\,^\omega$-semimodule pairs, just as the latter provided a generalization that united weighted languages with infinite languages for a subset of semirings.  Second, because people tend to be interested in systems with more than a single dimension, it would be useful to extend the formalism presented in this paper to power series over pictures (see Ref. \cite{Maurer2007}) which, like bidiverging power series, also have useful applications in quantum simulation (see Ref. \cite{Jordan2008}).

Finally, it is worth noting that we have demonstrated something very important here, which is that there is a significant link between automata theory and a family of techniques in quantum simulation.  It is a hope of the authors that this link will benefit both fields of research by leading to cross-fertilization of ideas between them.

\bibliographystyle{plain}
\bibliography{paper}

\end{document}